\documentclass[12pt,reqno]{amsart}
\usepackage{graphicx}
\usepackage{amssymb,amsmath}
\usepackage{amsthm}
\usepackage{color}
\usepackage[pdf]{pstricks}
\usepackage{hyperref}
\usepackage{empheq}
\usepackage{pgfplots}

\numberwithin{equation}{section}

\usepackage[makeroom]{cancel}

\usepackage{amssymb}
\usepackage{epsfig}
\usepackage{extarrows}
\usepackage{amsfonts}
\usepackage{graphicx}
\usepackage{caption}
\usepackage{subcaption}
\usepackage{tikz}
\usepackage{color}

\hypersetup{hypertex=true,
	colorlinks=true,
	linkcolor=blue,
	anchorcolor=g,
	citecolor=red}

\newcommand{\R}{\mathbb{R}}

\DeclareMathOperator{\sn}{sn}
\DeclareMathOperator{\cn}{cn}
\DeclareMathOperator{\dn}{dn}

\newtheorem{example}{Example}

\newtheorem{theorem}{Theorem}
\newtheorem{lemma}{Lemma}

\begin{document}
	
\title[Bright and dark breathers in the defocusing mKdV equation]{\bf Bright and dark breathers on an elliptic wave \\in the defocusing mKdV equation}

\author{Dmitry E. Pelinovsky}
\address[D. E. Pelinovsky]{Department of Mathematics and Statistics, McMaster University, Hamilton, Ontario, Canada, L8S 4K1}
\email{pelinod@mcmaster.ca}

\author{Rudi Weikard}
\address[R. Weikard]{Department of Mathematics, University of Alabama at Birmingham, Birmingham, AL 35294--1241 USA}
\email{weikard@uab.edu}

\maketitle

\begin{abstract}
	Breathers on an elliptic wave background consist of nonlinear superpositions of a soliton and a periodic wave, both traveling with different wave speeds and interacting periodically in the space-time. For the defocusing modified Korteweg--de Vries equation, the construction of general breathers has been an open problem since the elliptic wave is related to the elliptic degeneration of the hyperelliptic solutions of genus two. We have found a new representation of eigenfunctions of the Lax operator associated with the elliptic wave, which enables us to solve this open problem and to construct two families of breathers with bright (elevation) and dark (depression) profiles. 
\end{abstract}

\section{Introduction}

We consider the integrable model given by the defocusing mKdV (modified Korteweg--de Vries) equation written in the form:
\begin{equation}
\label{mkdv}
u_t-6u^2u_x+u_{xxx}=0,
\end{equation}
where $(x,t)\in \mathbb{R}\times\mathbb{R}$ and $u=u(x,t) \in \R$. As is well-known (see reviews in \cite{El,Hoefer}), there exists a three-parameter family of traveling periodic wave solutions of the defocusing mKdV equation (\ref{mkdv}), which we call {\em the elliptic wave}. In the symmetric case, it 
reduces to the snoidal profile, which is related 
to the genus-one elliptic solutions of integrable equations of the AKNS hierarchy \cite{B}. In the non-symmetric case, it is related to the elliptic degeneration of the genus-two hyperelliptic solutions. Although 
Riemann theta functions degenerate into Jacobi theta functions in this limit \cite{FayBook}, the general theory is too abstract for practical applications 
of the mKdV equation. 

The mKdV equation is one of the fundamental models of the dispersive hydrodynamics, which brings together the balance between the cubic dispersion 
and the cubic nonlinearity, with many applications in the dynamics of internal waves in seas and oceans \cite{Grimshaw,Pelinovsky}. In the dispersive 
hydrodynamics, the elliptic waves and breathers on their backgrounds capture dynamics of dispersive shock waves and soliton gases \cite{ACEHL}. {\em Breathers} are nonlinear superpositions of a soliton and a periodic wave, both traveling with different wave speeds and interacting periodically in the space-time. In the symmetric case, breathers with the dark (depression) profiles were constructed in \cite{MP24} by using the explicit expression for eigenfunctions of the Lax operator associated with the genus-one elliptic functions from \cite{Takahashi}. In the non-symmetric case, breathers with the kink profiles were constructed in \cite{AP25} for the only case when the spectral parameter is at the origin (which is the center of symmetry of the Lax spectrum). The main purpose of this work is to construct breathers with both bright (elevation) and dark (depression) profiles in the general case of the genus-two elliptic functions and for the arbitrary non-zero values of the spectral parameter. 

The mKdV equation (\ref{mkdv}) is related to the Lax system of linear equations written for the eigenfunction $\varphi \in C^2(\mathbb{R} \times \mathbb{R},\mathbb{C}^2)$ and the spectral parameter $\zeta \in \mathbb{C}$:
\begin{equation}
\label{LS}
\partial_x \varphi = U(\zeta,u) \varphi, \quad  
\partial_t \varphi = V(\zeta,u) \varphi, 
\end{equation}
where 
\begin{align*}
U(\zeta,u) &=\left(\begin{array}{ll} i\zeta & u\\ u & -i\zeta\end{array} \right), \\
V(\zeta,u)  &= \left( \begin{array}{ll} 4i\zeta^3+2i\zeta u^2 & 4\zeta^2u-2i\zeta u_x+2u^3-u_{xx}\\ 4\zeta^2u+2i\zeta u_x+2u^3-u_{xx} & -4i\zeta^3-2i\zeta u^2 \end{array}\right).
\end{align*}
Classical solutions of the mKdV equation (\ref{mkdv}) arise as a compatibility condition of the Lax system (\ref{LS}) given by $\partial_x \partial_t \varphi = \partial_t \partial_x \varphi$. Since the spectral problem 
\begin{equation}
\label{spectral-problem}
\left(\begin{array}{ll} -i \partial_x & i u\\ -i u & i \partial_x \end{array} \right) \varphi = \zeta \varphi,
\end{equation}
is self-adjoint in $L^2(\mathbb{R},\mathbb{C}^2)$, the admissible values of the spectral parameter $\zeta$ (forming the Lax spectrum) belong to a subset of $\mathbb{R}$. We construct explicit eigenfunctions $\varphi \in C^2(\mathbb{R} \times \mathbb{R},\mathbb{C}^2)$ in the case when $u(x,t) = \phi(x+ct)$ is a general elliptic wave with the (left-propagating) wave speed $c \in \mathbb{R}$ and the smooth periodic profile $\phi(x) : \mathbb{R}\to \mathbb{R}$ and when $\zeta \in \mathbb{R}$ is arbitrary.

\subsection{Main results} 

The periodic wave profile $\phi$ of the traveling wave $u(x,t) = \phi(x+ct)$ 
with the wave speed $c \in \mathbb{R}$ satisfies the third-order equation 
\begin{equation}
\label{first}
\phi'''-6\phi^2\phi'+c\phi' = 0,
\end{equation}
which can be integrated twice with two real constants of integration 
$b$ and $d$ as follows:
\begin{align}
\label{second}
\phi'' - 2 \phi^3 + c \phi &= b,  \\
\label{third}
( \phi')^2 - \phi^4 + c\phi^2 &= 2b \phi + 2 d.
\end{align}
The general periodic solution of the system (\ref{first}), (\ref{second}), and (\ref{third}) has different analytic forms for $b = 0$  and $b \neq 0$. 
For $b = 0$ ({\em symmetric case}), the level curves of (\ref{third}) on the phase plane $(\phi,\phi')$ are symmetric with respect to reflection $\phi \to -\phi$. For $b \neq 0$ ({\em non-symmetric case}), the symmetry of reflection $\phi \to -\phi$ is lost on the phase plane $(\phi,\phi')$. If $b \neq 0$ (without loss of generality, we consider $b > 0$), parameters ($b,c,d$) can be parameterized by the real parameters $(\zeta_1,\zeta_2,\zeta_3)$ satisfying $0 < \zeta_3 < \zeta_2 < \zeta_1$:
\begin{equation}
\label{parameterization}
\begin{cases}
b = 4 \zeta_1 \zeta_2 \zeta_3, \\
c = 2(\zeta_1^2 + \zeta_2^2 + \zeta_3^2), \\
d = \frac{1}{2}(\zeta_1^4 + \zeta_2^4 + \zeta_3^4) - \zeta_1^2 \zeta_2^2 - \zeta_1^2 \zeta_3^2 - \zeta_2^2 \zeta_3^2.
\end{cases}
\end{equation}
The elliptic wave is then represented in the standard form 
(see, e.g., \cite{Hoefer}):
\begin{equation}
\label{form-2}
\phi(x) = \dfrac{2 (\zeta_1 + \zeta_3)(\zeta_2 + \zeta_3)}{(\zeta_1 + \zeta_3) -(\zeta_1 - \zeta_2){\rm sn}^2(\nu x,k)} - \zeta_1 - \zeta_2 - \zeta_3,
\end{equation}
with
\begin{equation}
\label{parameters-nu-k}
\nu = \sqrt{\zeta_1^2 - \zeta_3^2},  \quad
k = \sqrt{\dfrac{\zeta_1^2 - \zeta_2^2}{\zeta_1^2 - \zeta_3^2}}.
\end{equation} 
Parameters $\{ \zeta_1,\zeta_2,\zeta_3\}$ in (\ref{parameterization}), (\ref{form-2}), and (\ref{parameters-nu-k}) define the end points of the Lax spectrum 
\begin{equation}
\label{Lax-spectrum}
\sigma_L = (-\infty,-\zeta_1] \cup [-\zeta_2,-\zeta_3] \cup [\zeta_3,\zeta_2] \cup [\zeta_1,\infty),
\end{equation}
with two symmetric bandgaps $(-\zeta_1,-\zeta_2) \cup (\zeta_2,\zeta_1)$ and the central bandgap $(-\zeta_3,\zeta_3)$. If $0 < \zeta_3 < \zeta_2 < \zeta_1$, then the wave profile $\phi$ is periodic on $\mathbb{R}$ with the period $2\nu^{-1} K(k)$ and on $i \mathbb{R}$ with the period $2 \nu^{-1} K'(k)$,  where $K(k)$ is a complete elliptic integral of the first kind and $K'(k) = K(k')$ with $k' = \sqrt{1-k^2}$. The spectral stability of the elliptic wave with the profile (\ref{form-2}) is shown 
in \cite{DN1} based on the Lax spectrum (\ref{Lax-spectrum}) and the 
squared eigenfunction relation between eigenfunctions of the Lax system (\ref{LS}) 
and the linearized mKdV equation.

In what follows, we drop the dependence of the elliptic functions and integrals on $k \in (0,1)$ if it does not cause a confusion. We also use standard notations for elliptic functions, which are reviewed in (\ref{Jacobi-theta}) and (\ref{Jacobi-sn}) below. 

The main result of this study is the representation formula for eigenfunctions of the Lax system (\ref{LS}) with $u(x,t) = \phi(x+ct)$ given by (\ref{form-2}) and with arbitrary $\zeta \in \mathbb{R}$. To do so, we use a 
new representation of the elliptic function $\phi$ obtained in \cite{AP25}, which is defined by poles $\pm \nu^{-1} (iK'+\alpha)$ and zeros $\pm \nu^{-1} \beta$, 
where the values of  $\alpha \in (0,K)$ and $\beta \in [0,K) \cup i[0,K')$ are uniquely obtained from 
\begin{equation}
\label{parameters-nu-k-alpha}
\sn(\alpha) = \sqrt{\dfrac{\zeta_1 - \zeta_3}{\zeta_1 + \zeta_2}}, \quad 
\sn(\beta) = \sqrt{\frac{(\zeta_1 + \zeta_3)(\zeta_1 - \zeta_2 - \zeta_3)}{
		(\zeta_1 - \zeta_2) (\zeta_1 + \zeta_2 + \zeta_3)}}.
\end{equation}
If $\zeta_1 \neq \zeta_2 + \zeta_3$, the wave profile $\phi$
is represented in the factorized form (see \cite[Theorem 1]{AP25}):
\begin{equation}
\label{gen-theta}
\phi(x) = (\zeta_1 - \zeta_2 - \zeta_3) \frac{\Theta^2(\alpha)}{H^2(\beta)} \frac{H(\nu x-\beta) H(\nu x + \beta)}{\Theta(\nu x-\alpha) \Theta(\nu x+\alpha)},
\end{equation}
where $\beta \in (0,K)$ for $\zeta_1 > \zeta_2 + \zeta_3$ and $\beta \in i(0,K')$ for $\zeta_1 < \zeta_2 + \zeta_3$. If $\zeta_1 = \zeta_2 + \zeta_3$, then $\beta = 0$, and the wave profile $\phi$ can be written in the factorized form:
\begin{equation}
\label{gen-theta-degenerate}
\zeta_1 = \zeta_2 + \zeta_3 : \quad 
\phi(x) = \frac{2 (\zeta_2 + \zeta_3) \zeta_3 \Theta^2(\alpha) \sn^2(\nu x)  \Theta^2(\nu x)}{(\zeta_2 + 2 \zeta_3) \Theta^2(0) \Theta(\nu x-\alpha) \Theta(\nu x+\alpha)}.
\end{equation}

\begin{figure}[htb!]
	\includegraphics[width=0.6\textwidth,height=0.3\textheight]{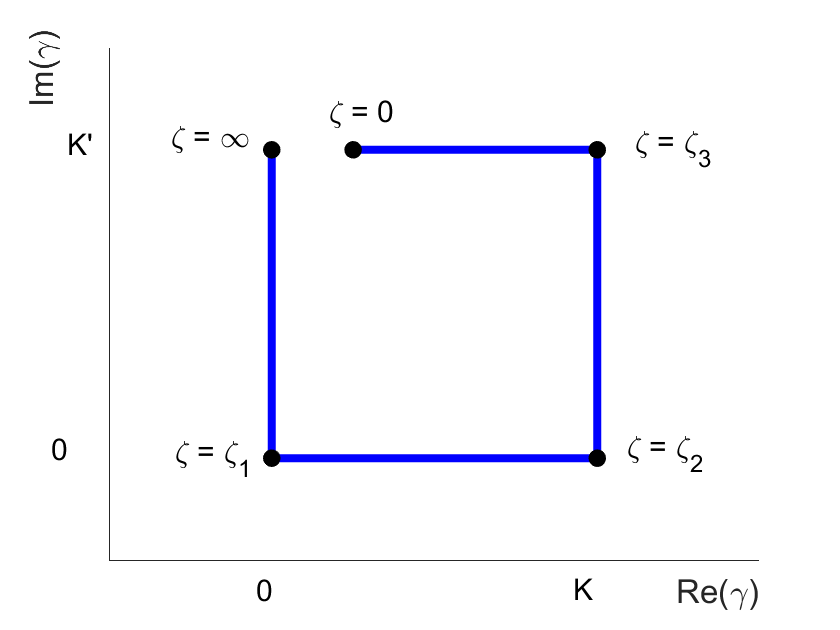}
	\caption{The pre-image of the mapping $[0,K]  \times [0,iK'] \ni \gamma \to \zeta \in [0,\infty)$ on the complex plane when $\zeta$ changes from $\zeta = 0$ to $\zeta = \infty$.}
	\label{fig-image}
\end{figure}

To introduce the eigenfunctions of the Lax system (\ref{LS}), we parameterize 
the spectral parameter $\zeta \in \mathbb{R}$ by using $\gamma \in [0,K] \times [0,iK']$ from the dispersion relation:
\begin{equation}
\label{dispersion}
\zeta^2 = \zeta_3^2 + (\zeta_1^2 - \zeta_3^2) \dn^2(\gamma).
\end{equation}
This dispersion relation is obtained in Lemma \ref{lem-gamma} below. Figure \ref{fig-image} shows the pre-image of the mapping 
$[0,K]  \times [0,iK'] \ni \gamma \to \zeta \in [0,\infty)$, 
which is justified in Lemma \ref{lem-beta} below. With this parameterization, the following theorem specifies the explicit representation of the eigenfunctions. 
The proof is based on Lemmas \ref{lem-coefficients} and \ref{lem-mu} below. 

\begin{theorem}
	\label{th-main}
	Let $u(x,t) = \phi(x+ct)$ be defined by (\ref{form-2}) for $0 < \zeta_3 < \zeta_2 < \zeta_1$ and $\gamma \in [0,K]  \times [0,iK']$ be defined by (\ref{dispersion}) for a given $\zeta \in (0,\zeta_3) \cup (\zeta_3,\zeta_2) \cup (\zeta_2,\zeta_1) \cup (\zeta_1,\infty)$. There exist two linearly independent solutions of the Lax system (\ref{LS}) in the form:
	\begin{equation}
	\label{eigenfunction-time}
	\varphi(x,t) = \left( \begin{matrix} p_1(x+ct)  \\ q_1(x+ct) \end{matrix} \right) e^{\mu t}, \quad \varphi(x,t) = \left( \begin{matrix} p_2(x+ct)  \\ q_2(x+ct) \end{matrix} \right) e^{-\mu t},
	\end{equation}
	where
	\begin{equation}
	\label{Lame-sol-1-fin}
	\left( \begin{matrix} p_1(x) \\ q_1(x) \end{matrix} \right) = 
	\left( \begin{matrix}  -i\zeta\frac{H(\nu x + \gamma + \alpha)}{\Theta(\nu x + \alpha)} + \zeta_1 \frac{\Theta(0) \Theta(2\alpha + \gamma)}{\Theta(\gamma) \Theta(2\alpha)} \frac{H(\nu x + \gamma - \alpha)}{\Theta(\nu x - \alpha)} \\
	-i\zeta\frac{H(\nu x + \gamma + \alpha)}{\Theta(\nu x + \alpha)} - \zeta_1 \frac{\Theta(0) \Theta(2\alpha + \gamma)}{\Theta(\gamma) \Theta(2\alpha)} \frac{H(\nu x + \gamma - \alpha)}{\Theta(\nu x - \alpha)} \end{matrix} \right) 
	e^{-\nu x Z(\gamma)}
	\end{equation}
	and
	\begin{equation}
	\label{Lame-sol-2-fin}
	\left( \begin{matrix} p_2(x) \\ q_2(x) \end{matrix} \right) = 
	\left( \begin{matrix}  \zeta_1 \frac{\Theta(0) \Theta(2\alpha + \gamma)}{\Theta(\gamma) \Theta(2\alpha)}\frac{H(\nu x - \gamma + \alpha)}{\Theta(\nu x + \alpha)} +  i\zeta \frac{H(\nu x - \gamma - \alpha)}{\Theta(\nu x - \alpha)} \\ 
	\zeta_1 \frac{\Theta(0) \Theta(2\alpha + \gamma)}{\Theta(\gamma) \Theta(2\alpha)}\frac{H(\nu x - \gamma + \alpha)}{\Theta(\nu x + \alpha)} -  i\zeta \frac{H(\nu x - \gamma - \alpha)}{\Theta(\nu x - \alpha)} \end{matrix} \right) e^{\nu x Z(\gamma)},
	\end{equation}
with
	\begin{align}
	\label{omega}
	\mu = -4 \nu^3 k^2 \sn(\gamma) \cn(\gamma) \dn(\gamma).  
	\end{align}
\end{theorem}

\begin{example}
	When $\zeta = 0$, we have $\gamma = 2 \alpha + i K'$, see Figure \ref{fig-image}. By using 
	\begin{equation}
	\label{Wei-Jac-3}
	H(z+iK') = i e^{\frac{\pi K'}{4 K}} e^{-\frac{i \pi z}{2K}} \Theta(z), \quad 
	\Theta(z+iK') = i e^{\frac{\pi K'}{4 K}} e^{-\frac{i \pi z}{2K}} H(z),
	\end{equation}
	expressions (\ref{Lame-sol-1-fin}) and (\ref{Lame-sol-2-fin}) reduce up to the norming constants to 
	\begin{equation}
	\label{exact-2}
	\left( \begin{matrix} p_1(x) \\ q_1(x) \end{matrix} \right) = 
	\left( \begin{matrix} 1 \\ -1 \end{matrix} \right) \frac{\Theta(\nu x + \alpha)}{\Theta(\nu x - \alpha)} e^{-\nu x \frac{H'(2\alpha)}{H(2\alpha)}}
	\end{equation}
	and
	\begin{equation}
	\label{exact-1}
	\left( \begin{matrix} p_2(x) \\ q_2(x) \end{matrix} \right) = 
	\left( \begin{matrix} 1 \\ 1 \end{matrix} \right) \frac{\Theta(\nu x - \alpha)}{\Theta(\nu x + \alpha)} 
	e^{\nu x \frac{H'(2\alpha)}{H(2\alpha)}},
	\end{equation}
	which were derived in \cite{AP25}. Furthermore, by using 
	\begin{equation}
	\label{rel-2}
	\sn(z \pm i K') = \frac{1}{k \sn(z)}, \;\; \cn(z \pm i K') = \frac{\mp i \dn(z)}{k \sn(z)}, \;\; 
	\dn(z \pm i K') = \frac{\mp i \cn(z)}{\sn(z)},
	\end{equation}
	and
	\begin{equation}
	\label{elliptic-double-alpha}
	\sn(2\alpha) = \frac{\nu}{\zeta_1}, \quad \cn(2\alpha) = \frac{\zeta_3}{\zeta_1}, \quad \dn(2\alpha) = \frac{\zeta_2}{\zeta_1},
	\end{equation}
	which follows from 
	\begin{equation}
	\label{elliptic-alpha}
	\sn(\alpha) = \frac{\sqrt{\zeta_1 - \zeta_3}}{\sqrt{\zeta_1 + \zeta_2}}, \quad 
	\cn(\alpha) = \frac{\sqrt{\zeta_2 + \zeta_3}}{\sqrt{\zeta_1 + \zeta_2}}, \quad 
	\dn(\alpha) = \frac{\sqrt{\zeta_2 + \zeta_3}}{\sqrt{\zeta_1 + \zeta_3}},
	\end{equation} 
	expression (\ref{omega}) yields $\mu = 4 \zeta_1 \zeta_2 \zeta_3$, also in agreement with \cite{AP25}.
\end{example}

Expressions (\ref{Lame-sol-1-fin}) and (\ref{Lame-sol-2-fin}) of Theorem \ref{th-main} will now be used to construct breathers on the elliptic wave for  eigenvalues $\zeta$ chosen in the gaps $(-\zeta_1,-\zeta_2)$, $(-\zeta_3,\zeta_3)$, and $(\zeta_2,\zeta_1)$ of the Lax spectrum (\ref{Lax-spectrum}).

\subsection{Construction of breathers}

Darboux transformation for the mKdV equation are well-known in the literature, e.g., see \cite{Gu,Matveev} or \cite[Appendix]{CP2018}.
The one-fold Darboux transformation is written in the form:
\begin{equation}
\label{DT}
u = \phi - \frac{4i \zeta p q}{p^2 - q^2},
\end{equation}
where $p = c_1 p_1 + c_2 p_2$ and $q = c_1 q_1 + c_2 q_2$ for some constant coefficients $c_1,c_2 \in \mathbb{C}$ and for a given choice of the spectral parameter $\zeta \in \mathbb{R}$. Expanding yields 
\begin{equation}
\label{DT-expand}
u = \phi - \frac{4i \zeta (c_1^2 p_1 q_1 + c_1 c_2 (p_1 q_2 + p_2 q_1) + c_2^2 p_2 q_2)}{c_1^2 (p_1^2 - q_1^2) + 2c_1 c_2 (p_1 p_2 - q_1 q_2) + c_2^2 (p_2^2 - q_2^2)}.
\end{equation}

The main application of Theorem \ref{th-main} is an algorithm on how to choose $\zeta \in \R$ and $c_1,c_2 \in \mathbb{C}$ in order to obtain bounded real-valued solutions $u = u(x,t)$ of the mKdV equation (\ref{mkdv}) from the transformation formula (\ref{DT-expand}), which describe breathers on the elliptic wave background. We show in Section \ref{sec-4} below that if $\zeta \in (0,\zeta_3) \cup (\zeta_2,\zeta_1)$, then we can choose scaling factors in $(p_1,q_1)$ and $(p_2,q_2)$ such that 
\begin{equation}
\label{condition-1}
p_1^2 - q_1^2, p_2^2 - q_2^2, p_1q_2 + p_2 q_1 \in i \mathbb{R}
\end{equation}
and 
\begin{equation}
\label{condition-2}
p_1p_2 - q_1 q_2, p_1 q_1, p_2 q_2 \in \mathbb{R}.
\end{equation}
This implies that the solution $u = u(x,t)$ in (\ref{DT-expand}) is real-valued if $c_1 c_2 \in i \mathbb{R}$, $c_1^2,c_2^2 \in \mathbb{R}$. The exponential factors in (\ref{Lame-sol-1-fin}) and (\ref{Lame-sol-2-fin}) are defined differently between $(0,\zeta_3)$ and $(\zeta_2,\zeta_1)$ by the parameter $\kappa > 0$:
\begin{equation}
\label{kappa}
\kappa = \left\{ \begin{array}{ll} 
	\frac{\Theta'(\gamma)}{\Theta(\gamma)}, \quad & \zeta \in (\zeta_2,\zeta_1), \\
	\frac{H'(\gamma - iK')}{H(\gamma - iK')}, \quad & \zeta \in (0,\zeta_3). \end{array} \right.
\end{equation}
It follows from (\ref{DT-expand}) that only the factor $c_2/c_1$ affects the representation of the solutions. To reduce the number of parameters, we introduce the phase translation of the breather relative to the elliptic wave and define
\begin{equation}
\label{c1-c2}
c_1 = e^{-\nu \kappa x_0 }, \quad 
c_2 = -i e^{\nu \kappa x_0 },
\end{equation}
where the negative sign of $c_2$ was chosen to obtain bounded solutions. 
The choice of the sign of $c_2$ is the only result based on 
the numerical verifications. All other results 
have been justified by using elliptic functions.

It follows from (\ref{eigenfunction-time}) 
with (\ref{Lame-sol-1-fin}) and (\ref{Lame-sol-2-fin}) that the breather  
propagates along the coordinate $\eta = x + c_s t + x_0$ with 
the (left-propagating) wave speed $c_s$ given by 
\begin{equation}
\label{speed}
c_s = c - \frac{\mu}{\nu \kappa}.
\end{equation}
In comparison with the three-parameter family of the elliptic wave 
with the profile $\phi$ in (\ref{form-2}), the solution obtained 
from (\ref{DT}), (\ref{DT-expand}), and (\ref{c1-c2}) has two additional parameters $\zeta \in (0,\zeta_3) \cup (\zeta_2,\zeta_1)$ (which determines uniquely $\gamma \in [0,K] \times [0,iK']$, see Figure \ref{fig-image}) and $x_0 \in \mathbb{R}$. Since $\kappa > 0$ in (\ref{kappa}), it follows from (\ref{Lame-sol-1-fin}) and (\ref{Lame-sol-2-fin}) that the elliptic wave is shifted across the breather with the phase shift in $x$ given by $2 \nu^{-1} (\gamma - \alpha)$.

To illustrate the breather solutions, let us first show the periodic profiles of $\phi$ in Figure \ref{Fig:Profile} for $(\zeta_1,\zeta_2,\zeta_3) = (2,1,0.5)$ (left) and 
$(\zeta_1,\zeta_2,\zeta_3) = (1.25,1,0.5)$ (right). Either solution form  (\ref{form-2}) or  (\ref{gen-theta}) can be used for plotting within the machine precision error since the two different representations (\ref{form-2}) and  (\ref{gen-theta}) are mathematically equivalent. The left panel corresponds to $\beta \in (0,K)$, for which $\phi$ has zeros on $\mathbb{R}$, and the right panel corresponds to $\beta \in i (0,K')$, for which $\phi$ has no zeros on $\mathbb{R}$.

\begin{figure}[htb!]
	\includegraphics[width=0.45\textwidth]{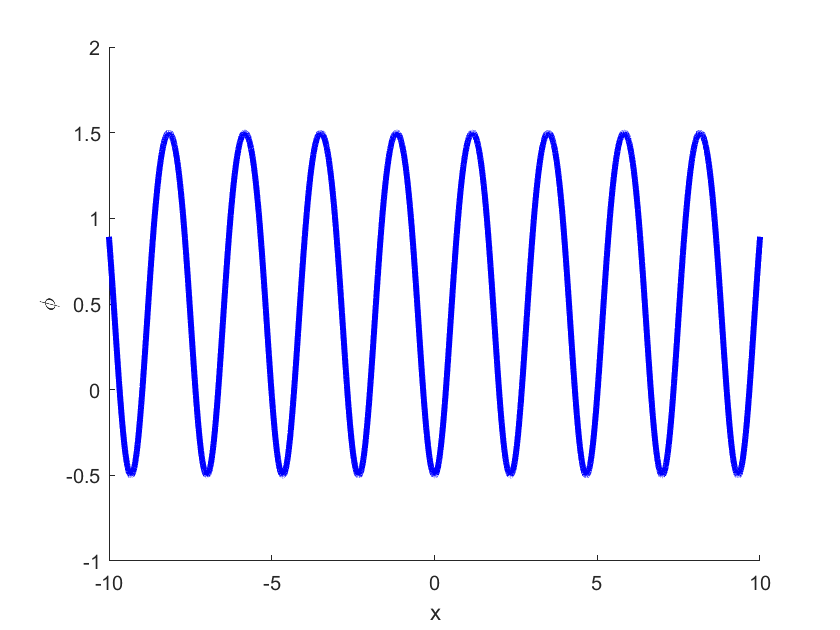}
	\includegraphics[width=0.45\textwidth]{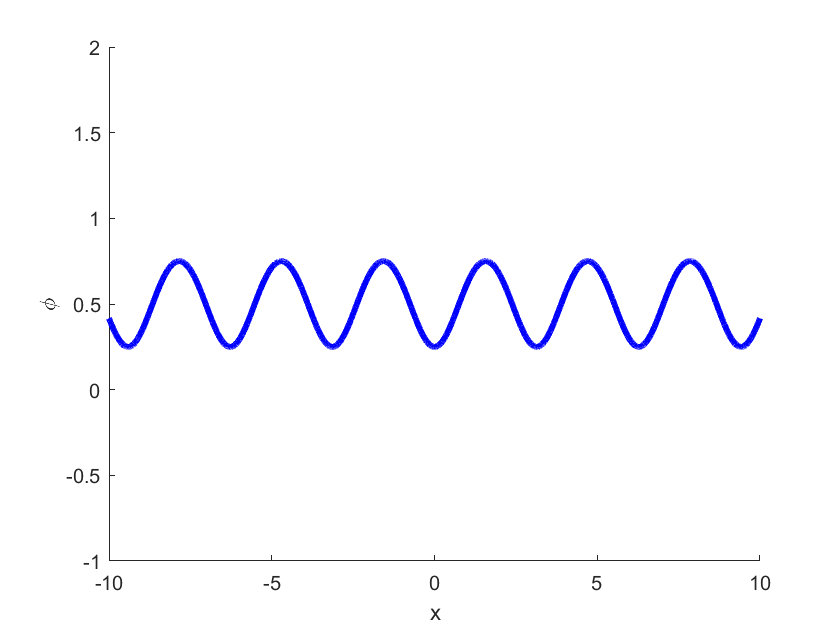}
	\caption{The elliptic function $\phi$ versus $x$ given by either (\ref{form-2}) or (\ref{gen-theta}) for $(\zeta_1,\zeta_2,\zeta_3) = (2,1,0.5)$ (left) and $(\zeta_1,\zeta_2,\zeta_3) = (1.25,1,0.5)$ (right).}
	\label{Fig:Profile}
\end{figure}

The solution $u = u(x,t)$ in (\ref{DT-expand}) for the choice of $\zeta$ in $(0,\zeta_3)$ is shown for $(\zeta_1,\zeta_2,\zeta_3) = (2,1,0.5)$ in Figure \ref{Fig:Breather2} and for $(\zeta_1,\zeta_2,\zeta_3) = (1.25,1,0.5)$ in Figure \ref{Fig:Breather12}. In both cases, the left panels show the snapshot versus $x$ at $t = 0$ and the right panels show the solution surface versus $(\xi,t)$, 
where $\xi = x + ct$. In both cases, we have set $x_0 = 0$ in (\ref{c1-c2}) 
and chosen $\zeta$ from (\ref{dispersion}) with $\gamma = 2 \alpha + 0.1 (K - 2 \alpha) + i K'$. It is obvious from Figures \ref{Fig:Breather2} and \ref{Fig:Breather12} that the solution $u = u(x,t)$ represents a breather propagating on the elliptic wave  background. For the choice of $\zeta$ in $(0,\zeta_3)$, the breather has the elevation (bright) profile and it propagates to the right relative to the traveling wave, which is stationary in the reference frame $\xi = x + ct$. This corresponds to $c_s < c$ in (\ref{speed}). 

In the limit $\gamma \to 2 \alpha + i K'$ (when $\zeta \to 0$), the breather profile becomes wide and the solution degenerates to the kink breather constructed in \cite[Theorem 2]{AP25}. There are two kinks by the symmetry: from negative to positive values of the periodic wave and from positive to negative values of the periodic wave. The kinks also moves to the right relative to the traveling wave, see Figure 4 in \cite{AP25}.

\begin{figure}[htb!]
	\includegraphics[width=0.45\textwidth]{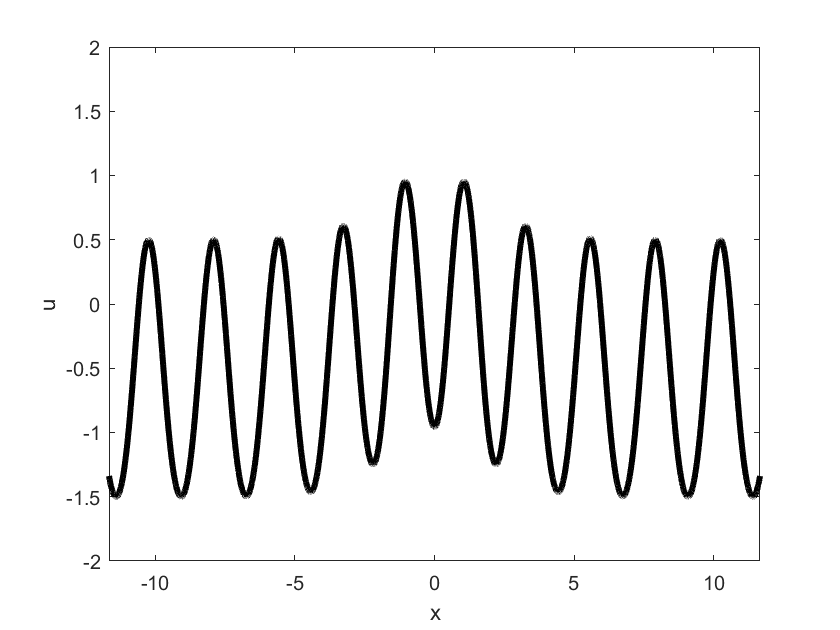}
	\includegraphics[width=0.45\textwidth]{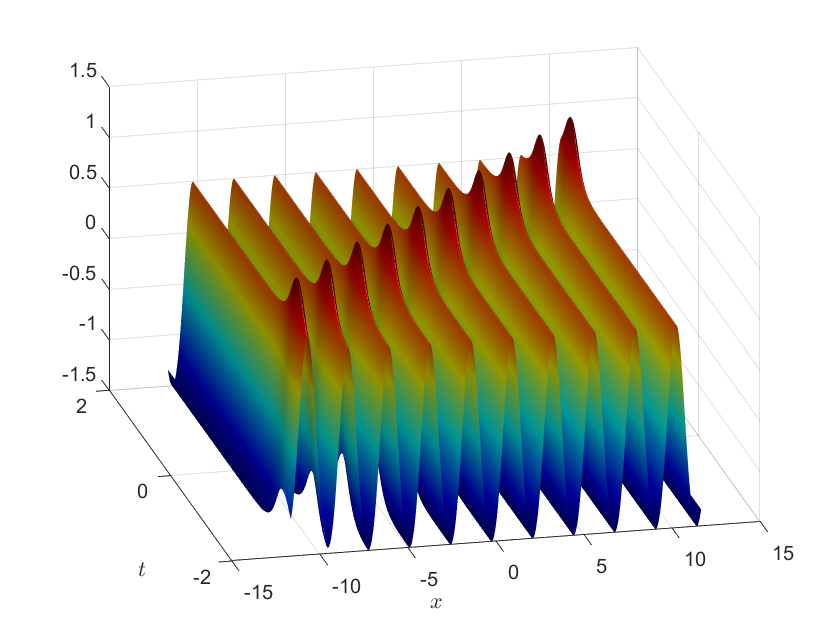}
	\caption{The snapshot versus $x$ for $t = 0$ (left) and the solution surface versus $(\xi,t)$ (right) for the bright breather with 
		$(\zeta_1,\zeta_2,\zeta_3) = (2,1,0.5)$. }
	\label{Fig:Breather2}
\end{figure}

\begin{figure}[htb!]
	\includegraphics[width=0.45\textwidth]{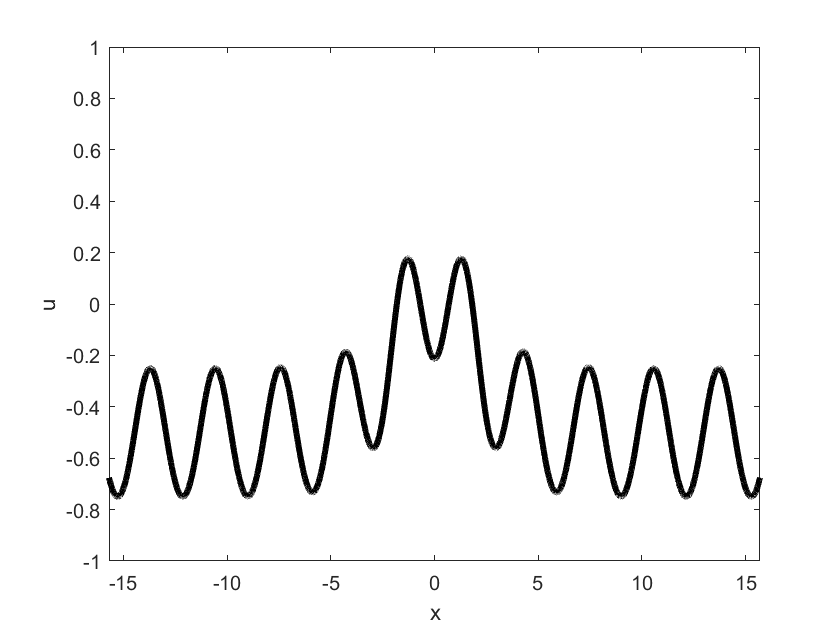}
	\includegraphics[width=0.45\textwidth]{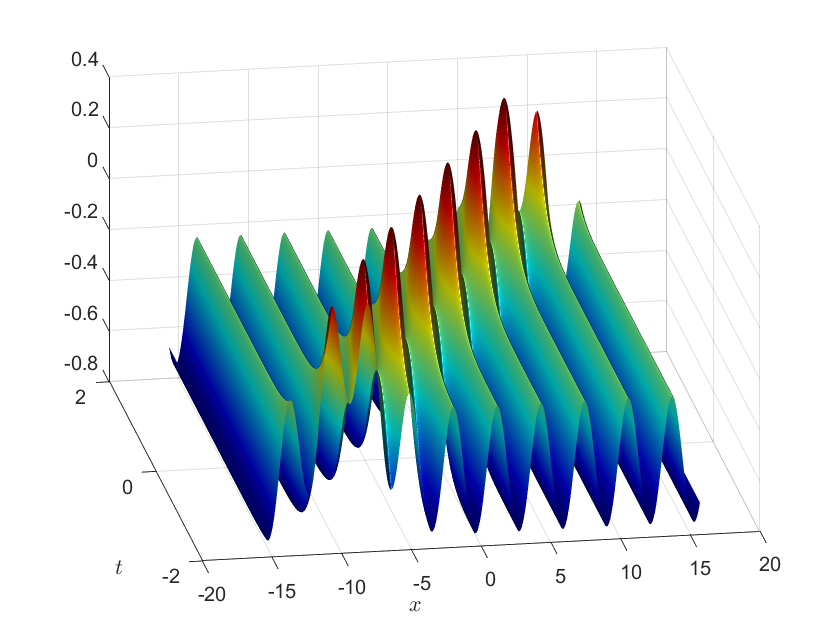}
	\caption{The same as Figure \ref{Fig:Breather2} but for  
		$(\zeta_1,\zeta_2,\zeta_3) = (2,1,0.5)$.}
	\label{Fig:Breather12}
\end{figure}

The solution $u = u(x,t)$ in (\ref{DT-expand}) for the choice of $\zeta$ in $(\zeta_2,\zeta_1)$ is shown for $(\zeta_1,\zeta_2,\zeta_3) = (2,1,0.5)$ in Figure \ref{Fig:Breather2-dark} and for $(\zeta_1,\zeta_2,\zeta_3) = (1.25,1,0.5)$ in Figure \ref{Fig:Breather12-dark}. The spectral parameter $\zeta$ is defined by (\ref{dispersion}) with $\gamma = 0.5 K$. It is again obvious from Figures \ref{Fig:Breather2-dark} and \ref{Fig:Breather12-dark} that the solution $u = u(x,t)$ represents a breather propagating on the elliptic wave  background. For the choice of $\zeta$ in $(\zeta_2,\zeta_1)$, the breather has the depression (dark) profile and it propagates  to the left relative to the traveling wave, which is stationary in the reference frame $\xi = x + ct$. This coresponds to $c_s > c$ in (\ref{speed}). 

In the limit $\zeta_3 \to 0$, the elliptic wave is reduced to the snoidal profile by the Landen transformation \cite[Example 3.2]{AP25}. Dark 
breathers on the snoidal profile were constructed explicitly in \cite[Theorem 1]{MP24}. The dark breathers also move to the left relative to the traveling wave, see Figure 1 in \cite{MP24}.

\begin{figure}[htb!]
	\includegraphics[width=0.45\textwidth]{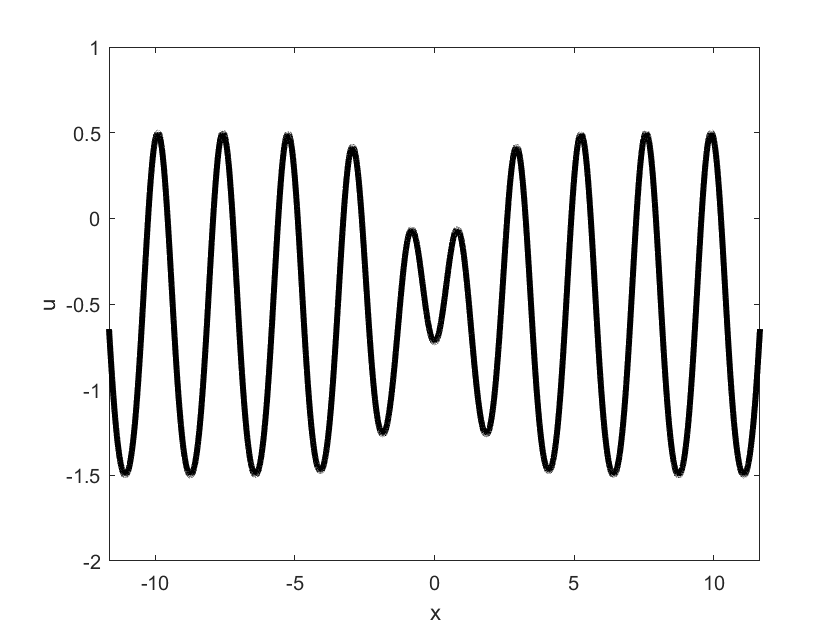}
	\includegraphics[width=0.45\textwidth]{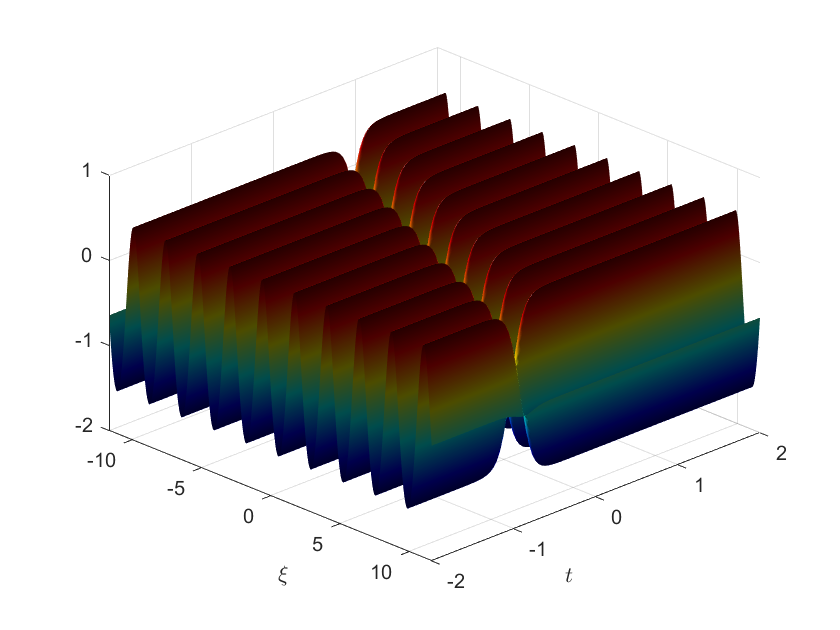}
	\caption{The snapshot versus $x$ for $t = 0$ (left) and the solution surface versus $(\xi,t)$ (right) for the dark breather with 
		$(\zeta_1,\zeta_2,\zeta_3) = (2,1,0.5)$. }
	\label{Fig:Breather2-dark}
\end{figure}

\begin{figure}[htb!]
	\includegraphics[width=0.45\textwidth]{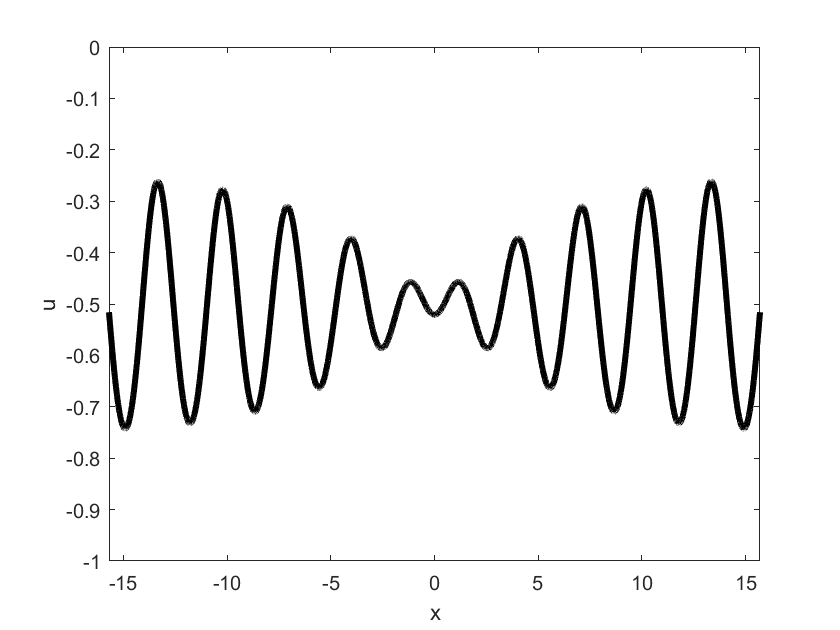}
	\includegraphics[width=0.45\textwidth]{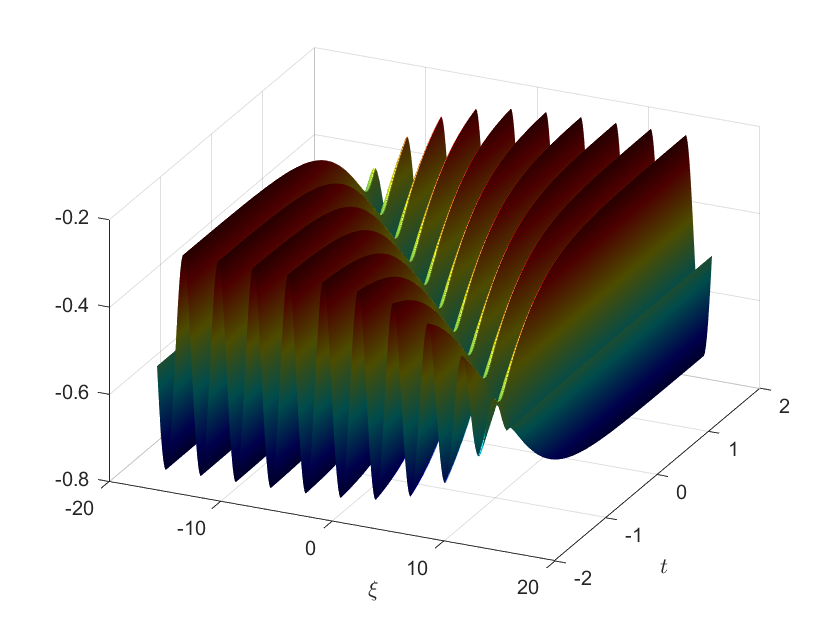}
	\caption{The same as Figure \ref{Fig:Breather2-dark} but for  
		$(\zeta_1,\zeta_2,\zeta_3) = (2,1,0.5)$.}
	\label{Fig:Breather12-dark}
\end{figure}

\subsection{Background of the problem}

Breathers for the Korteweg--de Vries (KdV) equation were first constructed in the pioneering work \cite{KM74} based on the dressing method. Algebro-geometric constructions of breathers were developed in \cite{GS95} with a rather abstract approach based on commutation identities. Another method of constructing such solutions for the KP hierarchy was developed in \cite{BEN20,Nak} with the degeneration of Grassmanians, see also \cite{LiZhang}.
 
Inspired by the recent studies of the soliton gas in \cite{Congy,Girotti,Grava}, the breather solutions on the elliptic wave background appear to be the central objects for interactions of solitons in the limit of infinitely many solitons. Consequently, the same solutions of the KdV equation were comprehensively studied in two independent works \cite{Bertola,HMP23}. It is shown in \cite{HMP23} that breathers with the bright profiles are related to the spectral parameter chosen below the Lax spectrum and breathers with the dark profiles are related to the choice of the spectral parameter in a spectral gap of a finite length. 

The same bright and dark profiles of breathers are found in the nonlocal Benjamin--Ono (BO) equation \cite{ChenPel24}, where the traveling wave is expressed in terms of trigonometric rather than elliptic functions but the Lax spectrum and the choice of the spectral parameter are similar to the case of the KdV equation. For the nonlocal derivative nonlinear Schr\"{o}dinger NLS equation, it is shown in \cite{ChenPel25} that only dark profiles exist in the defocusing case, whereas both bright and dark profiles coexist in the focusing case. Similarly to the BO equation, the traveling wave is still expressed by trigonometric functions and has a very similar Lax spectrum but the spectral parameter for the breathers can be chosen inside the continuous spectrum, contrary to the cases of KdV and BO equations. 

For the defocusing cubic NLS equation, dark profiles of breathers 
were constructed for the snoidal elliptic wave in  \cite{Shin,Takahashi}, see also \cite{Ling}. It was found recently in \cite{Sun} 
that both bright and dark profiles of breathers exist for the elliptic 
wave with a non-trivial phase, which is still expressed by the genus-one elliptic 
functions.  The elliptic wave is modulationally unstable for the focusing cubic NLS equation and formations of rogue waves and time-decaying and space-periodic breathers complicate dynamics of perturbations \cite{Biondini2,CPW,Feng}.

For the defocusing mKdV equation, the results of \cite{AP25} and this study shows that breathers propagating on the general elliptic wave may have bright, dark, and kink profiles. This is very different to the only dark profiles found on the snoidal elliptic wave in the symmetric case \cite{MP24}.
For the focusing mKdV equation, one family of traveling waves (generalizing the dnoidal profiles) is modulationally stable and the other family of traveling waves (generalizing the cnoidal profiles) is modulationally unstable, see \cite{Cui} and the references therein. Consequently, bright breathers were observed on the dnoidal profiles \cite{Grava,Minakov} and rogue waves were constructed on the cnoidal 
profiles \cite{CP2018,CP2019}, but the comprehensive study of 
breathers in the focusing mKdV equation is still open.

For the focusing mKdV equation, stability of the elliptic waves and breathers on their background were only developed in the symmetric case \cite{LingSun,LingSun2}. The main obstacle has been again 
the explicit characterization of the eigenfunctions of the Lax operators. Only 
recently in \cite{LingSun3}, the eigenfunctions are constructed by degeneration of the Riemann theta functions of genus two, very similar to the general elliptic theory in \cite{FayBook}. One of the possible application of our work is to construct the explicit eigenfunctions of the Lax operators for the focusing mKdV equation in terms of the elliptic functions. We anticipate that 
the eigenfunctions can be obtained in a similar form to the one in Theorem \ref{th-main} but the challenges arise in analysis of the pre-image of the mapping $[0,K]  \times [0,iK'] \ni \gamma \to \zeta \in \sigma_L$ since the Lax spectrum 
$\sigma_L$ is no longer a subset of $\R$.

Finally, we mention another form used to characterize eigenfunctions 
of the Lax operators related to the elliptic degeneration of the Riemann theta functions of genus two \cite{Geng,Geng3}. The elliptic functions can be written in the factorized form given by as quotient of products of Jacobi theta functions, 
similarly to the forms (\ref{gen-theta}) and (\ref{gen-theta-degenerate}). The factorized form for the eigenfunctions of the Lax system (\ref{LS}) are expected to be useful in the construction of breathers on the general elliptic wave by using the transformation formula (\ref{DT}), as it was the case for the snoidal wave in \cite{MP24}. However, we show in Section \ref{sec-5} below that the poles and zeros of the elliptic eigenfunctions in the factorized form depend on the spectral parameter via the branch point singularities, and therefore, the factorized form is not as explicit as the representation of eigenfunctions obtained in Theorem \ref{th-main}.

\subsection{Methods and organization of the paper}

Theorem \ref{th-main} is proven in Sections \ref{sec-2} and \ref{sec-3}. 
The main idea for the proof comes from an application of the Miura transformation 
\begin{equation}
\label{Miura-time}
v(x,t) = u^2(x,t) \pm u_x(x,t),
\end{equation}
which relates solutions $u(x,t)$ of the mKdV equation (\ref{mkdv}) to solutions  $v(x,t)$ of the KdV equation $v_t + 6 v v_x + v_{xxx} = 0$. The Miura 
transformation can be used to map the spectral problem (\ref{spectral-problem}) to 
the stationary Schr\"{o}dinger equation. By using Weierstrass' elliptic functions, we show that the Schr\"{o}dinger equation for the general elliptic wave (\ref{form-2}) becomes Lam\'{e} equation with a single-gap elliptic potential, solutions of which are well-known, see \cite[p.395]{Ince} and \cite[Section 3]{GW95}. This yields the representation of the eigenfunctions of the spectral problem 
(\ref{spectral-problem}) as a linear combination of two solutions of the 
Schr\"{o}dinger equation. The coefficients of this linear combination are 
uniquely computed by using expansions of elliptic functions near poles in the complex plane. This is achieved in Section \ref{sec-2} with the complex parameter $\gamma \in [0,K] \times [0,iK']$ to represent the eigenfunctions in the form (\ref{Lame-sol-1-fin}) and (\ref{Lame-sol-2-fin}). 

The pre-image of the mapping $[0,K] \times [0,iK'] \ni \gamma \to \zeta \in [0,\infty)$ shown in Figure \ref{fig-image} is computed in Section \ref{sec-3}, where we also compute the time-dependent form (\ref{eigenfunction-time}) of eigenfunctions. The mapping $[0,K] \times [0,iK'] \ni \gamma \to \mu \in \mathbb{C}$ is obtained in the explicit form (\ref{omega}) also by using 
relations to Weierstrass' elliptic functions and computation at a single point $x_0 \in \mathbb{C}$. We also show that $\mu \in \mathbb{R}$ if $\zeta \in (-\zeta_1,-\zeta_2) \cup (-\zeta_3,\zeta_3) \cup (\zeta_2,\zeta_1)$, which are 
the three gaps of the Lax spectrum (\ref{Lax-spectrum}) associated with the elliptic wave. 

Breathers in the form (\ref{DT-expand}) are studied in Section \ref{sec-4}, 
where we justify the conditions (\ref{condition-1}), (\ref{condition-2}), and (\ref{c1-c2}) for the breather solution to be 
real-valued and bounded. For the positive half-gap $(0,\zeta_3)$, where bright breathers of Figures \ref{Fig:Breather2} and \ref{Fig:Breather12} are constructed, we only need to rewrite the eigenfunctions of Theorem \ref{th-main} 
for $\gamma = iK' + \delta$ with $\delta \in (2\alpha,K)$ and choose the norming constants appropriately. However, for the positive gap $(\zeta_2,\zeta_1)$, where dark breathers of Figures \ref{Fig:Breather2-dark} and \ref{Fig:Breather12-dark} are constructed, we use $\gamma \in (0,K)$ but change $x$ to $iK' + x$. This transformation maps the bounded profile $\phi$ of the elliptic wave in (\ref{form-2}) to the singular profile but also maps the singular profile of the breather solution $u = u(x,t)$ in (\ref{DT-expand}) into a bounded profile.

Finally, in Section \ref{sec-5}, we discuss the factorized form for elliptic eigenfunctions of the Lax system (\ref{LS}) associated with the elliptic wave. We show that the poles and zeros of the elliptic eigenfunctions depend on the spectral parameter via the branch pole singularities, which cannot be unfolded with the use of elliptic functions. 

\subsection{Notations for elliptic functions}

We use Jacobi's theta functions:
 \begin{align*}
\left\{ \begin{array}{l} 
\theta_1(y) = 2 \sum\limits_{n=1}^{\infty} (-1)^{n-1} q^{(n-\frac{1}{2})^2} \sin(2n-1) y, \\
\theta_4(y) = 1 + 2 \sum\limits_{n=1}^{\infty} (-1)^{n} q^{n^2} \cos 2n y,
\end{array} \right.
\end{align*}
where $q := e^{-\frac{\pi K'}{K}}$ with $K$ and $K'$ being the complete elliptic integrals for elliptic modulus $k$ and $k' = \sqrt{1-k^2}$, respectively. For notational convenience, we use 
\begin{equation}
\label{Jacobi-theta}
H(x) = \theta_1(y), \quad \Theta(x) = \theta_4(y), \quad \mbox{\rm with} \quad y = \frac{\pi x}{2 K}.
\end{equation}
Jacobi's theta functions (\ref{Jacobi-theta}) are related to the elliptic function $\sn$, $\cn$, and $\dn$ due to
\begin{equation}
\label{Jacobi-sn}
\sn(x) = \frac{H(x)}{\sqrt{k} \Theta(x)}, \quad \cn(x) =  \frac{\sqrt{k'} H(x+K)}{\sqrt{k} \Theta(x)}, \quad \dn(x) = \frac{\sqrt{k'} \Theta(x + K)}{\Theta(x)}.
\end{equation}
These elliptic functions are related by the fundamental relations 
\begin{equation}
\label{fund-rel}
\sn^2(x) + \cn^2(x) = 1, \quad \dn^2(x) + k^2 \sn^2(x) = 1, \quad \dn^2(x) - k^2 \cn^2(x) = (k')^2.
\end{equation}
Functions $H$, $\sn$, and $\cn$ are $2K$-antiperiodic and $2iK'$-periodic, whereas functions $\Theta$ and $\dn$ are $2K$-periodic and $2iK'$-periodic. Further properties of elliptic functions can be consulted in \cite{Lawden}.

\section{Eigenfunctions of the spectral problem (\ref{spectral-problem})}
\label{sec-2}

We derive here the explicit representations (\ref{Lame-sol-1-fin}) and (\ref{Lame-sol-2-fin}) of elliptic eigenfunctions in Theorem \ref{th-main}.  
To do so, we first give the relation between the elliptic wave (\ref{form-2}) and Weierstrass' elliptic functions (Section \ref{sec-2-1}), which is 
used in the derivation. Then, we proceed with the Miura transformation (\ref{Miura-time}) to derive the explicit solutions of the Lam\'{e} equation with a single-gap 
elliptic potential (Section \ref{sec-2-2}). Finally, we obtain the unique 
expressions for the linear superpositions of explicit solutions 
of the Lam\'{e} equation from expansion of elliptic functions near poles 
in the complex plane (Section \ref{sec-2-3}).

\subsection{Relation to Weierstrass' elliptic function}
\label{sec-2-1}

We recall from \cite{A1990,W}, reviewed in \cite{AP25}, that the elliptic wave (\ref{form-2}) is related to Weierstrass' elliptic function 
\begin{equation}
\label{rel-Jac-Wei}
\wp(x) = e_3 + \frac{\nu^2}{\sn^2(\nu x,k)}, \quad \nu = \sqrt{e_1-e_3}, \quad k = \sqrt{\frac{e_2-e_3}{e_1-e_3}},
\end{equation}
where $e_3 < e_2 < e_1$ are given by the relations
\begin{equation}
\label{rel-par-e-zeta}
\left\{ \begin{array}{l} 
\displaystyle e_1 = \frac{1}{3} (\zeta_1^2 + \zeta_2^2 - 2 \zeta_3^2), \vspace{0.25cm} \\
\displaystyle e_2 = \frac{1}{3} (\zeta_1^2 - 2 \zeta_2^2 + \zeta_3^2), 
\vspace{0.25cm} \\
\displaystyle e_3 = \frac{1}{3} (-2 \zeta_1^2 + \zeta_2^2 + \zeta_3^2),
\end{array} \right.
\quad \Rightarrow \quad 
\left\{ \begin{array}{l} 
e_1 - e_2 = \zeta_2^2 - \zeta_3^2,
\vspace{0.1cm} \\
e_1 - e_3 = \zeta_1^2 - \zeta_3^2, 
\vspace{0.1cm}\\
e_2 - e_3 = \zeta_1^2 - \zeta_2^2.
\end{array} \right.
\end{equation}
Weierstrass' function $\wp$ is periodic with periods $2\omega$ and $2 \omega'$, where 
\begin{equation}
\label{periods}
\omega = \frac{K}{\sqrt{e_1 - e_3}}, \quad 
\omega' = \frac{i K'}{\sqrt{e_1 - e_3}}.
\end{equation}
By \cite[Lemma 3.4]{AP25}, there exists $v \in [-\omega,\omega] \times [-\omega',\omega']$ such that   
\begin{equation}\label{pcW}
\frac{c}{6}=\wp(v), \quad \frac{b}{2}=\wp'(v).
\end{equation}
By using $v \in [-\omega,\omega] \times [-\omega',\omega']$, the elliptic functions $\phi$, $\phi'$, and $\phi^2$ are related to Weierstrass' function 
$\wp$ as follows:
\begin{equation}
\label{phinova}
\phi(x) = \frac{1}{2}\frac{\wp'(x-\frac{v}{2})+\wp'(x+\frac{v}{2})}{\wp(x-\frac{v}{2})-\wp(x + \frac{v}{2})} = \zeta\left(x+\frac{v}{2}\right)-\zeta\left(x-\frac{v}{2}\right)-\zeta(v),
\end{equation}
\begin{equation}
\label{philinha}
\phi' (x) = \wp\left(x-\frac{v}{2} \right) - \wp\left(x+\frac{v}{2}\right),
\end{equation}
and
\begin{equation}
\label{form-3-squared}
\phi^2(x) = \wp\left(x + \frac{v}{2}\right) + \wp\left(x - \frac{v}{2}\right) + \wp(v).
\end{equation}
where $\zeta$ is Weierstrass' zeta function. Moreover, by \cite[Lemma 3.6]{AP25},
we have the correspondence between $v \in [-\omega,\omega] \times [-\omega',\omega']$ in (\ref{pcW}) and $\alpha \in (0,K)$ in (\ref{parameters-nu-k-alpha}) given by 
\begin{equation}
\label{v-alpha-correspondence}
\frac{v}{2} = -\frac{iK' + \alpha}{\sqrt{e_1-e_3}}.
\end{equation}

\subsection{Miura transformation}
\label{sec-2-2}

We consider here the spectral problem (\ref{spectral-problem}) with $u(x,t) = \phi(x+ct)$ and use $x$ in place of $x + ct$. Squaring the spectral problem (\ref{spectral-problem}) yields 
\begin{equation}
\label{sp-1}
\left( \begin{matrix} -\partial_x^2 + \phi^2 & \phi' \\
\phi' & -\partial_x^2 + \phi^2 \end{matrix} \right) \varphi = \zeta^2 \varphi,
\end{equation}
where  $\varphi = (p,q)^T $ is the eigenfunction and $\zeta \in \mathbb{R}$ is the spectral parameter. By folding (\ref{sp-1}) in variables $\psi_{\pm} = p \pm q$,  we get two uncoupled Schr\"{o}dinger equations 
\begin{equation}
\label{sp-2} 
\left( - \partial_x^2 + \phi^2 \pm \phi' \right) \psi_{\pm} = \zeta^2 \psi_{\pm}.
\end{equation}
Explicit solutions of the Schr\"{o}dinger equations (\ref{sp-2}) are described in the following lemma. 

\begin{lemma}
	\label{lem-gamma}
	Let $\gamma \in [0,K] \times [0,iK']$ be defined from the spectral parameter $\zeta \in \mathbb{R}$ by the dispersion relation
	\begin{equation}
	\label{Lame-disp}
	\zeta^2 = \zeta_1^2 - (\zeta_1^2 - \zeta_2^2) \sn^2(\gamma) = 
	\zeta_2^2 + (\zeta_1^2 - \zeta_2^2) \cn^2(\gamma) = 
	\zeta_3^2 + (\zeta_1^2 - \zeta_3^2) \dn^2(\gamma).
	\end{equation}
	Assume that $\gamma \neq \{ 0, K, iK', K+iK'\}$.
	General solutions to the Schr\"{o}dinger equations (\ref{sp-2}) are given by 
	\begin{equation}
	\label{Lame-sol}
	\psi_{\pm}(x) = 2 c_1^{\pm} \frac{H(\nu x + \gamma \pm \alpha)}{\Theta(\nu x \pm \alpha)} e^{-\nu x Z(\gamma)} + 
	2 c_2^{\pm} \frac{H(\nu x -\gamma \pm \alpha)}{\Theta(\nu x \pm \alpha)} e^{\nu x Z(\gamma)},
	\end{equation}
	where $c_1^{\pm},c_2^{\pm} \in \mathbb{C}$ are arbitrary coefficients and the factor $2$ is used for convenience.

\end{lemma}

\begin{proof}
By using (\ref{philinha}) and (\ref{form-3-squared}), we obtain 
\begin{align}
\label{Miura}
\phi^2(x) \pm \phi'(x) = \wp(v) + 2 \wp\left(x \mp \frac{v}{2} \right).
\end{align}
In view of (\ref{v-alpha-correspondence}), we obtain from (\ref{rel-Jac-Wei}) and (\ref{rel-par-e-zeta}) that 
\begin{align}
\wp\left(x \mp \frac{v}{2} \right) &= e_3 + \frac{e_1-e_3}{\sn^2(\sqrt{e_1-e_3}(x \mp \frac{v}{2}))} \notag \\
&= \frac{1}{3} (-2\zeta_1^2 + \zeta_2^2 + \zeta_3^2) + \frac{\zeta_1^2 - \zeta_3^2}{\sn^2(\nu x \pm i K' \pm \alpha)} \notag \\
&= \frac{1}{3} (-2\zeta_1^2 + \zeta_2^2 + \zeta_3^2) + (\zeta_1^2 - \zeta_2^2) \sn^2(\nu x \pm \alpha), \label{rel-1}
\end{align}
where we have used (\ref{rel-2}) as well as the definition of $\nu$ and $k$ in (\ref{parameters-nu-k}). Since 
\begin{equation}
\label{rel-3}
\wp(v) = \frac{c}{6} = \frac{1}{3} (\zeta_1^2 + \zeta_2^2 + \zeta_3^2),
\end{equation}
we finally obtain from (\ref{sp-2}), (\ref{Miura}), and (\ref{rel-1}):
\begin{equation}
\label{sp-3} 
\left( - \partial_x^2 -2 (\zeta_1^2-\zeta_2^2) \cn^2(\nu x \pm \alpha)  \right) \psi_{\pm} = (\zeta^2 - \zeta_1^2 + \zeta_2^2 - \zeta_3^2) \psi_{\pm}.
\end{equation}
Recall from \cite[p.395]{Ince} that the Lam\'{e} equation 
\begin{equation}
\label{Lame-eq}
-\varphi''(z) + 2 k^2 \sn^2(z) \varphi(z) = \eta \varphi(z) 
\end{equation}
admits for $\eta \neq \{ 1,k^2,1+k^2,\infty \}$ a general solution in the form: 
\begin{equation*}
\varphi(z) = c_1 \frac{H(z + \gamma)}{\Theta(z)} e^{-z Z(\gamma)} + 
c_2 \frac{H(z-\gamma)}{\Theta(z)} e^{z Z(\gamma)}, \quad Z(\gamma) := \frac{\Theta'(\gamma)}{\Theta(\gamma)}, 
\end{equation*}
where $c_1,c_2 \in \mathbb{C}$ are arbitrary coefficients and 
$\gamma \in [0,K] \times [0,iK']$ is defined from the spectral parameter $\eta \in \mathbb{R}$ by the dispersion relation $\eta = k^2 + \dn^2(\gamma)$. 
If $\eta \neq \{ 1,k^2,1+k^2,\infty \}$, then $\gamma \neq \{ 0, K, iK', K+iK'\}$.
Since $z = \nu x \pm \alpha$,  we obtain (\ref{Lame-disp}) and (\ref{Lame-sol}) by comparing (\ref{sp-3}) with (\ref{Lame-eq}) and using (\ref{parameters-nu-k}) for $$
\eta = 1 + \frac{\zeta^2 - \zeta_2^2}{\zeta_1^2 - \zeta_3^2}.
$$ 
The three relations in (\ref{Lame-disp}) are equivalent due to 
the fundamental relations for elliptic functions (\ref{fund-rel}).
\end{proof}

\subsection{Coefficients of the linear superposition}
\label{sec-2-3}

Since $\psi_{\pm} = p \pm q$, we get two linearly independent 
solutions $\varphi = (p_1,q_1)$ and $\varphi = (p_2,q_2)$ of the spectral problem 
(\ref{spectral-problem}) from solutions (\ref{Lame-sol}) by using the inverse transformation 
	$$
	p = \frac{\psi_+ + \psi_-}{2}, \quad q = \frac{\psi_+ - \psi_-}{2}.
	$$
Separating the two solutions in (\ref{Lame-sol}) according to the two exponential factors, we obtain the two linearly independent solutions in the form:
\begin{equation}
\label{Lame-sol-1}
\left( \begin{matrix} p_1(x) \\ q_1(x) \end{matrix} \right) = 
\left( \begin{matrix} 
c_1^+ \frac{ H(\nu x + \gamma + \alpha)}{\Theta(\nu x + \alpha)} + c_1^- \frac{H(\nu x + \gamma - \alpha)}{\Theta(\nu x - \alpha)} \\
c_1^+ \frac{ H(\nu x + \gamma + \alpha)}{\Theta(\nu x + \alpha)} - c_1^- \frac{H(\nu x + \gamma - \alpha)}{\Theta(\nu x - \alpha)} \end{matrix} 
\right)  e^{-\nu x Z(\gamma)}
\end{equation}
and 
\begin{equation}
\label{Lame-sol-2}
\left( \begin{matrix} p_2(x) \\ q_2(x) \end{matrix} \right) = 
\left( \begin{matrix} c_2^+ \frac{ H(\nu x - \gamma + \alpha)}{\Theta(\nu x + \alpha)} + c_2^- \frac{H(\nu x - \gamma - \alpha)}{\Theta(\nu x - \alpha)} \\
c_2^+ \frac{ H(\nu x - \gamma + \alpha)}{\Theta(\nu x + \alpha)} - c_2^- \frac{H(\nu x - \gamma - \alpha)}{\Theta(\nu x - \alpha)}  \end{matrix} \right) 
e^{\nu x Z(\gamma)}. 
\end{equation}
Since there exist only two linearly-independent solutions of the spectral problem (\ref{spectral-problem}), each set $(c_1^+,c_1^-)$ and $(c_2^+,c_2^-)$ cannot be arbitrary but must be linearly dependent.

\begin{example}
If $\zeta = 0$, we have $\gamma = 2 \alpha + i K'$, see Figure \ref{fig-image}.
Relations (\ref{Wei-Jac-3}) imply that 
$$
\frac{\Theta'(\gamma)}{\Theta(\gamma)} = \frac{\Theta'(2 \alpha + i K')}{\Theta(2 \alpha + i K')} = -\frac{i \pi}{2K} + \frac{H'(2\alpha)}{H(2\alpha)}.
$$
The expression (\ref{Lame-sol-1}) with $\gamma = 2 \alpha + i K'$ yields (\ref{exact-2})  if $c_1^+ = 0$ since 
$$
\frac{H(\nu x + \gamma - \alpha)}{\Theta(\nu x - \alpha)} e^{-\nu x \frac{\Theta'(\gamma)}{\Theta(\gamma)}} = 
\frac{H(\nu x + \alpha + iK')}{\Theta(\nu x - \alpha)} e^{-\nu x \frac{\Theta'(2 \alpha + i K')}{\Theta(2 \alpha + i K')}} = 
i e^{\frac{\pi K'}{4 K}} \frac{\Theta(\nu x + \alpha)}{\Theta(\nu x - \alpha)}
e^{-\nu x \frac{H'(2\alpha)}{H(2\alpha)}}.
$$
The expression (\ref{Lame-sol-2}) with $\gamma = 2 \alpha + i K'$ yields (\ref{exact-1}) if $c_2^- = 0$ since 
$$
\frac{H(\nu x - \gamma + \alpha)}{\Theta(\nu x + \alpha)} e^{\nu x \frac{\Theta'(\gamma)}{\Theta(\gamma)}} = 
\frac{H(\nu x - \alpha - iK')}{\Theta(\nu x + \alpha)} e^{\nu x \frac{\Theta'(2 \alpha + i K')}{\Theta(2 \alpha + i K')}} = 
-i e^{\frac{\pi K'}{4 K}} \frac{\Theta(\nu x - \alpha)}{\Theta(\nu x + \alpha)}
e^{\nu x \frac{H'(2\alpha)}{H(2\alpha)}}.
$$
The proportionality factors in both expressions are not relevant due to the scalar multiplication of eigenfunctions. Thus, if $\zeta = 0$, then $c_1^+ = 0$ and $c_2^- = 0$. 
\label{example-2}
\end{example}

The following lemma describes the linear relations in each set $(c_1^+,c_1^-)$ 
and $(c_2^+,c_2^-)$.

\begin{lemma}
	\label{lem-coefficients}
If the representations (\ref{Lame-sol-1}) and (\ref{Lame-sol-2}) are solutions of the spectral problem (\ref{spectral-problem}), then their coefficients are uniquely defined up to the constant multiplication factor as 
\begin{equation}
\label{coefficient-1}
	c_1^+ = -i \zeta, \quad c_1^- = \zeta_1 \frac{\Theta(0) \Theta(2\alpha + \gamma)}{\Theta(\gamma) \Theta(2\alpha)}
\end{equation}
and
\begin{equation}
\label{coefficient-2}
c_2^+ = \zeta_1 \frac{\Theta(0) \Theta(2\alpha + \gamma)}{\Theta(\gamma) \Theta(2\alpha)}, \quad c_2^- = i \zeta.
\end{equation}
\end{lemma}

\begin{proof}
To obtain (\ref{coefficient-1}), we compute the pole contributions of the elliptic solutions of the spectral problem (\ref{spectral-problem}) at 
$$
x = \frac{v}{2} = -\frac{\alpha + i K'}{\nu}.
$$ 
Since $\wp(x) = \frac{1}{x^2} + \tilde{\wp}(x)$ as $x \to 0$ with analytic  $\tilde{\wp}(x)$ near $x = 0$, it follows from (\ref{phinova}) that 
\begin{equation}
\label{phi-asymp}
\phi(x) = -\frac{1}{x - \frac{v}{2}} + \mathcal{O}\left(x - \frac{v}{2} \right), \quad \mbox{\rm as} \;\; x \to \frac{v}{2}.
\end{equation}
By using (\ref{Jacobi-sn}) for $H(x) = \sqrt{k} \sn(x) \Theta(x)$, we
obtain the asymptotic expansion:
\begin{equation*}
H(x) = \sqrt{k} \Theta(0) x + \mathcal{O}(x^3), \quad \mbox{\rm as} \;\; x \to 0.
\end{equation*}
In view of (\ref{Wei-Jac-3}) and (\ref{v-alpha-correspondence}), this implies that 
\begin{equation}
\label{Theta-asymp}
\Theta(\nu x + \alpha) = -i \nu \sqrt{k} \Theta(0) e^{\frac{\pi K'}{4K}} e^{\frac{i \pi \nu}{2K} \left(x - \frac{v}{2} \right)} \left( x - \frac{v}{2} \right) + \mathcal{O}\left(x - \frac{v}{2} \right)^3, \quad \mbox{\rm as} \;\; x \to \frac{v}{2}.
\end{equation}
By using (\ref{Lame-sol-1}) and (\ref{Theta-asymp}), we obtain as $x \to \frac{v}{2}$:
\begin{align*}
& \frac{H(\nu x + \gamma + \alpha)}{\Theta(\nu x + \alpha)} e^{-\nu x Z(\gamma)}\\
&= \frac{i}{\nu \sqrt{k} \Theta(0)} e^{-\frac{\pi K'}{4K}} 
e^{(\alpha + iK') Z(\gamma)} e^{-\nu \left(x - \frac{v}{2} \right) \left[ Z(\gamma) + \frac{i \pi}{2K}\right]}
\frac{H\left(\nu\left(x - \frac{v}{2}\right) + \gamma - iK' \right)}{x - \frac{v}{2}} + \mathcal{O}\left(x - \frac{v}{2}\right) \\
&= \frac{1}{\nu \sqrt{k} \Theta(0)} 
e^{(\alpha + iK') Z(\gamma) + \frac{i \pi \gamma}{2K}} e^{-\nu \left(x - \frac{v}{2} \right) Z(\gamma)}
\frac{\Theta\left(\nu\left(x - \frac{v}{2}\right) + \gamma \right)}{x - \frac{v}{2}} + \mathcal{O}\left(x - \frac{v}{2}\right)
\end{align*}
and
\begin{align*}
\frac{H(\nu x + \gamma - \alpha)}{\Theta(\nu x - \alpha)} e^{-\nu x Z(\gamma)}
&= e^{(\alpha + iK') Z(\gamma)} \frac{H(\gamma - 2 \alpha - iK')}{\Theta(-2\alpha - iK')} + \mathcal{O}\left(x - \frac{v}{2}\right) \\
&= e^{(\alpha + iK') Z(\gamma) + \frac{i \pi \gamma}{2K}} \frac{\Theta(\gamma - 2 \alpha)}{H(-2\alpha)} + \mathcal{O}\left(x - \frac{v}{2}\right).
\end{align*}
In order to substitute the meromorphic parts into $p' = i \zeta p + \phi q$, 
we deduce that 
\begin{align}
p(x) &=  \frac{c_1^+ }{ \nu \sqrt{k} \Theta(0)} 
e^{(\alpha + iK') Z(\gamma) + \frac{i \pi \gamma}{2K}} 
\frac{\Theta(\gamma)}{x - \frac{v}{2}} + \mathcal{O}(1), 
\quad \mbox{\rm as} \;\; x \to \frac{v}{2},
\label{p-asymp}
\end{align}
which yields
\begin{align}
p'(x) &= -\frac{c_1^+ }{ \nu \sqrt{k} \Theta(0)} 
e^{(\alpha + iK') Z(\gamma) + \frac{i \pi \gamma}{2K}} 
\frac{\Theta(\gamma)}{\left( x - \frac{v}{2} \right)^2} 
+ \mathcal{O}(1), 
\quad \mbox{\rm as} \;\; x \to \frac{v}{2}.
\label{p-der-asymp}
\end{align}
On the other hand, we get 
\begin{align}
q(x) &= \frac{c_1^+ }{\nu \sqrt{k}\Theta(0)} 
e^{(\alpha + iK') Z(\gamma) + \frac{i \pi \gamma}{2K}} e^{-\nu \left(x - \frac{v}{2} \right) Z(\gamma)}
\frac{\Theta\left(\nu\left(x - \frac{v}{2}\right) + \gamma \right)}{x - \frac{v}{2}} \notag \\
& \quad - c_1^- e^{(\alpha + iK') Z(\gamma) + \frac{i \pi \gamma}{2K}} \frac{\Theta(\gamma - 2 \alpha)}{H(-2\alpha)} 
+ \mathcal{O}\left(x - \frac{v}{2}\right),
\quad \mbox{\rm as} \;\; x \to \frac{v}{2}.
\label{q-asymp}
\end{align}
The double pole in $p'(x)$ given by (\ref{p-der-asymp}) 
cancels out with the double pole in $\phi(x) q(x)$ given by 
the product of (\ref{phi-asymp}) and (\ref{q-asymp}). The simple pole 
from $i \zeta p(x)$ given by (\ref{p-asymp}) must cancel 
the simple pole in $\phi(x) q(x)$, due to a relation 
between $c_1^+$ and $c_1^-$. Using (\ref{phi-asymp}), (\ref{p-asymp}), (\ref{p-der-asymp}), and (\ref{q-asymp}) in $p' = i \zeta p + \phi q$, 
we obtain 
\begin{equation*}
i \zeta \frac{c_1^+ \Theta(\gamma)}{\nu \sqrt{k} \Theta(0)} -  \frac{c_1^+ }{\sqrt{k} \Theta(0)} \left[ - \Theta(\gamma) Z(\gamma) + \Theta'(\gamma) \right] + c_1^- \frac{\Theta(\gamma - 2 \alpha)}{H(-2\alpha)} = 0, 
\end{equation*}
after cancelling the constant factor $e^{(\alpha + iK') Z(\gamma) + \frac{i \pi \gamma}{2K}}$. Since $ \Theta'(\gamma) = \Theta(\gamma) Z(\gamma)$, this yields 
the linear equation 
\begin{equation}
\label{rel-c-1}
i \zeta  \frac{c_1^+ \Theta(\gamma)}{\nu \sqrt{k} \Theta(0)} 
- c_1^- \frac{\Theta(2 \alpha - \gamma)}{H(2\alpha)} = 0.
\end{equation}
The same equation follows also from $q' = -i\zeta q + \phi p$ as $x \to \frac{v}{2}$. 

Repeating the derivation as $x \to -\frac{v}{2}$, where 
$$
\left\{ \begin{array}{l} 
\phi(x) = \frac{1}{x + \frac{v}{2}} + \mathcal{O}\left(x + \frac{v}{2}\right), \\
\Theta(\nu x - \alpha) = i \nu \sqrt{k}\Theta(0)  e^{\frac{\pi K'}{4K}} e^{-\frac{i \pi \nu}{2K} \left(x + \frac{v}{2} \right)} \left( x + \frac{v}{2} \right) + \mathcal{O}\left(x + \frac{v}{2} \right)^3, \end{array} \right. \quad \mbox{\rm as} \;\; x \to -\frac{v}{2}, 
$$
we obtain another linear equation 
\begin{equation}
\label{rel-c-2}
c_1^+ \frac{\Theta(2 \alpha + \gamma)}{H(2\alpha)} + i \zeta \frac{c_1^-  \Theta(\gamma)}{\nu \sqrt{k} \Theta(0)} = 0.
\end{equation}
The determinant of the linear system (\ref{rel-c-1}) and (\ref{rel-c-2}) is zero if 
\begin{equation}
\label{det-eq-from-lin-syst}
\zeta^2 = \frac{\nu^2 k \Theta^2(0) \Theta(2\alpha-\gamma) \Theta(2\alpha + \gamma)}{\Theta^2(\gamma) H^2(2\alpha)}.
\end{equation}
Since $H(x) = \sqrt{k} \sn(x) \Theta(x)$ and 
$$
\Theta^2(0) \Theta(2\alpha-\gamma) \Theta(2\alpha+\gamma) = \Theta^2(2\alpha) \Theta^2(\gamma) - H^2(2\alpha) H^2(\gamma), 
$$
we use the relation $\nu = \zeta_1 \sn(2\alpha)$ from (\ref{elliptic-double-alpha}) and obtain from (\ref{det-eq-from-lin-syst}):
\begin{align*}
\zeta^2 &= \nu^2 k \left[ \frac{\Theta^2(2 \alpha)}{H^2(2 \alpha)} - \frac{H^2(\gamma)}{\Theta^2(\gamma)} \right] \\
&= \nu^2 \left[ \frac{1}{\sn^2(2\alpha)} - k^2 \sn^2(\gamma) \right] \\
&= \zeta_1^2 -  (\zeta_1^2 - \zeta_2^2) \sn^2(\gamma),
\end{align*}
which recovers the dispersion relation (\ref{Lame-disp}). Therefore, one 
of the two equations in the linear system (\ref{rel-c-1}) and (\ref{rel-c-2}) is redundant.

Recall from Example \ref{example-2} that $\gamma = 2 \alpha + i K'$ for $\zeta = 0$, which yields $c_1^+ = 0$ since $\Theta(2\alpha - \gamma) = 0$ and $\Theta(2\alpha + \gamma) \neq 0$. Therefore, we use (\ref{rel-c-2}) of the linear system (\ref{rel-c-1}) and (\ref{rel-c-2}), set $c_1^+ = -i\zeta$ and obtain 
$$
c_1^- = \frac{\nu \sqrt{k} \Theta(0) \Theta(2\alpha + \gamma)}{\Theta(\gamma) H(2\alpha)} = \zeta_1 \frac{\Theta(0) \Theta(2\alpha + \gamma)}{\Theta(\gamma) \Theta(2\alpha)},
$$
where we have used again that $\nu = \zeta_1 \sn(2\alpha)$. This yields (\ref{coefficient-1}).

To obtain (\ref{coefficient-2}) with a similar method, we expand the second solution (\ref{Lame-sol-2}) as $x \to \pm \frac{v}{2}$.  Here, we use the fact that (\ref{Lame-sol-2}) 
is obtained from (\ref{Lame-sol-1}) by replacing $\gamma$ to $-\gamma$. Since the poles $\pm \frac{v}{2}$ are independent of $\gamma$, we obtain the linear system of equations for $(c_2^+,c_2^-)$ by replacing $\gamma$ to $-\gamma$ in (\ref{rel-c-1}) and (\ref{rel-c-2}):
\begin{equation}
\label{rel-c-3}
i \zeta  \frac{c_2^+ \Theta(\gamma)}{\nu \sqrt{k} \Theta(0)} 
- c_2^- \frac{\Theta(2 \alpha + \gamma)}{H(2\alpha)} = 0
\end{equation}
and
\begin{equation}
\label{rel-c-4}
c_2^+ \frac{\Theta(2 \alpha - \gamma)}{H(2\alpha)} + i \zeta \frac{c_2^-  \Theta(\gamma)}{\nu \sqrt{k} \Theta(0)} = 0.
\end{equation}
Again, it follows from Example \ref{example-2} for $\zeta = 0$ that $c_2^- = 0$ since $\Theta(2\alpha - \gamma) = 0$ and $\Theta(2\alpha + \gamma) \neq 0$. 
Therefore, we use (\ref{rel-c-3}) of the linear system 
(\ref{rel-c-3}) and (\ref{rel-c-4}), set $c_2^- = i\zeta$ and obtain 
$$
c_2^+ = \frac{\nu \sqrt{k} \Theta(0) \Theta(2\alpha + \gamma)}{\Theta(\gamma) H(2\alpha)} = \zeta_1 \frac{\Theta(0) \Theta(2\alpha + \gamma)}{\Theta(\gamma) \Theta(2\alpha)}.
$$
This yields (\ref{coefficient-2}).
\end{proof}

The first solution (\ref{Lame-sol-1}) with (\ref{coefficient-1}) yields 
the final form (\ref{Lame-sol-1-fin}) in Theorem \ref{th-main}, whereas  
the second solution  (\ref{Lame-sol-2}) with (\ref{coefficient-2}) yields 
the final form (\ref{Lame-sol-2-fin}).

\section{Time evolution of the eigenfunctions}
\label{sec-3}

Having proven the representation (\ref{Lame-sol-1-fin}) and (\ref{Lame-sol-2-fin}) in Section \ref{sec-2}, we derive here the time evolution (\ref{eigenfunction-time}) of the elliptic eigenfunctions stated in Theorem \ref{th-main}. To do so, we first inspect the relationship between the spectral parameter $\zeta \in \mathbb{R}$ and the shift parameter $\gamma \in [0,K] \times [0,iK']$ given by the dispersion relation (\ref{dispersion}) (Section \ref{sec-3-1}). Then, we separate variables in the Lax system (\ref{LS}) and obtain the characteristic polynomial 
for the traveling waves (Section \ref{sec-3-2}), which is also solved by the elliptic functions. Finally, we derive the unique expression (\ref{omega}) for $\mu$ in terms of the shift parameter $\gamma$ (Section \ref{sec-3-3}).

\subsection{Dispersion relation for eigenfunctions}
\label{sec-3-1}

The following lemma defines the pre-image of the interval $[0,\infty)$ 
under the mapping $\gamma \to \zeta$ given by the dispersion relation (\ref{dispersion}). The pre-image is shown 
as the path in the complex $\gamma$ plane in Figure \ref{fig-image}.

\begin{lemma}
	\label{lem-beta}
	The pre-image of the path $\infty \to \zeta_1 \to \zeta_2 \to \zeta_3 \to 0$ for $\zeta \in (0,\infty)$ is given by the path $(0,iK') \to (0,0) \to (K,0) \to (K,iK') \to (2\alpha,iK')$ in the complex $\gamma$-plane.
\end{lemma}

\begin{proof}
	For $\zeta \in [\zeta_1,\infty)$, we use $\gamma = i \gamma'$ with $\gamma' \in [0,K']$ and rewrite (\ref{dispersion}) as 
	\begin{equation}
	\label{zeta-1}
	\zeta^2 = \zeta_3^2 + (\zeta_1^2 - \zeta_3^2) \frac{\dn^2(\gamma';k')}{\cn^2(\gamma';k')},
	\end{equation}
	where $k' = \sqrt{1-k^2}$. Since $\cn^2(K';k') = 0$, the image of $\gamma = iK'$ is $\zeta = \infty$. Since $\cn^2(0;k') = \dn^2(0;k') = 1$, the image of $\gamma = 0$ is $\zeta = \zeta_1$. For $\gamma' \in (0,K')$, we have $\zeta \in (\zeta_1,\infty)$. \\
	
	For $\zeta \in [\zeta_2,\zeta_1]$, we use $\gamma \in [0,K]$ in (\ref{dispersion}), rewritten again as 
	\begin{equation}
	\label{zeta-12}
	\zeta^2 = \zeta_3^2 + (\zeta_1^2 - \zeta_3^2) \dn^2(\gamma),
	\end{equation}
	that the image of $\gamma = 0$ is $\zeta = \zeta_1$ and the image of $\gamma = K$ is $\zeta = \zeta_2$ since $\zeta^2 = \zeta_3^2 + (\zeta_1^2 - \zeta_3^2) (1-k^2) = \zeta_2^2$. For $\gamma \in (0,K)$, we have $\zeta \in (\zeta_2,\zeta_1)$. \\ 
	
	For $\zeta \in [\zeta_3,\zeta_2]$, we use $\gamma = K + i \gamma'$ with $\gamma' \in [0,K']$ and rewrite (\ref{dispersion}) as 
	\begin{equation}
	\label{zeta-2}
	\zeta^2 = \zeta_3^2 + (\zeta_1^2 - \zeta_3^2) \frac{1-k^2}{\dn^2(i\gamma';k)} = \zeta_3^2 + (\zeta_2^2 - \zeta_3^2) \frac{\cn^2(\gamma';k')}{\dn^2(\gamma';k')}.
	\end{equation}
	Since $\cn^2(0;k') = \dn^2(0;k') = 1$, the image of $\gamma = K$ is $\zeta = \zeta_2$. Since $\cn^2(K';k') = 0$, the image of $\gamma = K + i K'$ is $\zeta = \zeta_3$. For $\gamma' \in (0,K')$, we have $\zeta \in (\zeta_3,\zeta_2)$. \\ 
	
	For $\zeta \in [0,\zeta_3]$, we use $\gamma = \delta + i K'$ with $\gamma \in [0,K]$ and rewrite (\ref{dispersion}) as 
	\begin{equation}
	\label{zeta-3}
	\zeta^2 = \zeta_3^2 - (\zeta_1^2 - \zeta_3^2) \frac{\cn^2(\delta)}{\sn^2(\delta)}.
	\end{equation}
	Since $\cn^2(K) = 0$, the image of $\gamma = K + iK'$ is $\zeta^2 = \zeta_3^2$. By (\ref{elliptic-double-alpha}) and $\nu^2 = \zeta_1^2 - \zeta_3^2$, the image of $\gamma = 2\alpha + i K'$ is $\zeta = 0$. For $\delta \in (0,K)$, we have $\zeta \in (0,\zeta_3)$. 	
\end{proof}

Expressions (\ref{zeta-1}) and (\ref{zeta-2}) are useful in the spectral bands $[\zeta_1,\infty)$ and $[\zeta_3,\zeta_2]$, whereas expressions (\ref{zeta-12}) and (\ref{zeta-3}) are useful in the spectral gaps $(\zeta_2,\zeta_1)$ and $[0,\zeta_3)$. We will only use the latter expressions for the purpose of this work. The following lemma gives a useful relation between the dispersion relation (\ref{dispersion}) and Weierstrass' function $\wp$. 

\begin{lemma}
	\label{lem-Weierstrass}
	Let $a \in [0,\omega] \times [0,\omega']$ be defined by 
	\begin{equation}
	\label{relation-gamma-a}
	\gamma = i K' + \nu a.
	\end{equation}
	The dispersion relation (\ref{dispersion}) can be rewritten as 
\begin{equation}
\label{Lame-disp-W}
\zeta^2 = \wp(v) - \wp(a).
\end{equation}
\end{lemma}

\begin{proof}
By using (\ref{rel-Jac-Wei}), (\ref{rel-par-e-zeta}), and (\ref{rel-3}), we obtain 
\begin{align*}
\zeta^2 &= \wp(v) - \wp(a) \\
&= \frac{1}{3} (\zeta_1^2 + \zeta_2^2 + \zeta_3^2) - \frac{1}{3} (-2\zeta_1^2 + \zeta_2^2 + \zeta_3^2) - \frac{\zeta_1^2 - \zeta_3^2}{\sn^2(\gamma \mp i K')} \\
&= \zeta_1^2 - (\zeta_1^2 - \zeta_2^2) \sn^2(\gamma),
\end{align*}
which is equivalent to (\ref{dispersion}).
\end{proof}

\subsection{Characteristic polynomial for traveling waves}
\label{sec-3-2}

The variables $(x,t)$ in the Lax system (\ref{LS}) with the traveling wave $u(x,t) = \phi(x+ct)$ can be separated in the form 
$\varphi(x,t) = \psi(x+ct) e^{\mu t}$. With the separation of variables, the Lax system (\ref{LS}) is split into the spectral problem
\begin{equation}
\label{spectral}
\frac{d}{dx} \psi = \left(\begin{array}{ll} i\zeta & \phi \\ \phi  & -i\zeta\end{array} \right) \psi
\end{equation}
and the linear algebraic system
\begin{align}
\label{time}
\mu \psi + c \left(\begin{array}{ll} i\zeta & \phi \\ \phi  & -i\zeta\end{array} \right) \psi = \left( \begin{array}{ll} 4i\zeta^3+2i\zeta \phi^2 & 4\zeta^2 \phi - 2i\zeta \phi' +2\phi^3 - \phi'' \\ 4\zeta^2 \phi + 2i \zeta \phi' + 2 \phi^3 - \phi'' & -4i\zeta^3-2i\zeta \phi^2 \end{array}\right) \psi,
\end{align}
where all functions depend on only one variable which stands for the traveling wave coordinate $x+ct$. 

The following lemma gives the admissible values of $\mu$, where we use Lemma \ref{lem-Weierstrass} and a relation to Weierstrass' function $\wp$.

\begin{lemma}
\label{lem-polynomial}
Let $a \in [0,\omega] \times [0,\omega']$ be defined by (\ref{relation-gamma-a}) and (\ref{Lame-disp-W}). Then, $\mu^2 = 4 (\wp'(a))^2$.
\end{lemma}

\begin{proof}
Since (\ref{time}) is the linear algebraic system, the admissible values of $\mu$ are found from the characteristic equation 
\begin{equation}
\label{char-eq}
\left| 
\begin{array}{ll} 4i\zeta^3+2i\zeta \phi^2 - ic \zeta - \mu & 4\zeta^2 \phi - 2i\zeta \phi' + 2\phi^3 - \phi'' - c \phi \\ 4\zeta^2 \phi + 2i \zeta \phi' + 2 \phi^3 - \phi'' - c\phi & -4i\zeta^3-2i\zeta \phi^2 + i c \zeta - \mu \end{array} 
\right| = 0, 
\end{equation}
whereas the quotient $\rho = q/p$ for the eigenfunction $\psi = (p,q)^T$ satisfies  
\begin{equation}
\label{rho-basic}
\rho = -\frac{4 i \zeta^3 + i \zeta (2 \phi^2 - c) - \mu}{4 \zeta^2 \phi - 2 i \zeta \phi' - b} = \frac{4 \zeta^2 \phi + 2 i \zeta \phi' - b}{4 i \zeta^3 + i \zeta (2 \phi^2 - c) + \mu}.
\end{equation}
Expanding the determinant in (\ref{char-eq}) and using (\ref{second}) and (\ref{third}), we obtain
\begin{equation}
\label{character-eq}
\mu^2 + 16 P(\zeta) = 0,
\end{equation}
where the characteristic polynomial $P(\zeta)$ can be written in the form:
\begin{align*}
P(\zeta) &= \zeta^6 - \frac{c}{2} \zeta^4 + \frac{1}{16} (c^2 - 8d) \zeta^2 - \frac{b^2}{16} \\
&= (\zeta^2 - \zeta_1^2) (\zeta^2 - \zeta_2^2) (\zeta^2 - \zeta_3^2),
\end{align*}
due to the parameterization (\ref{parameterization}). 
Substituting (\ref{rel-3}) and (\ref{Lame-disp-W}) into (\ref{character-eq}) and using (\ref{rel-par-e-zeta}), we obtain 
\begin{align*}
\mu^2 &= 16 (\zeta_1^2 + \wp(a) - \wp(v)) (\zeta_2^2 + \wp(a) - \wp(v)) (\zeta_3^2 + \wp(a) - \wp(v)) \\
&= 16 (\wp(a) - e_1) (\wp(a) - e_2) (\wp(a) - e_3) \\
&= 4 (\wp'(a))^2,
\end{align*}
where we used the first-order quadrature $(\wp')^2 = 4 (\wp - e_1) (\wp - e_2) (\wp - e_3)$ for Weierstrass' elliptic function $\wp$. 
\end{proof}

\subsection{Explicit expression for $\mu$}
\label{sec-3-3}

	By Lemma \ref{lem-polynomial}, there exist two admissible roots for $\mu$ 
for the two linearly independent solutions of the spectral problem (\ref{spectral}). The following lemma shows that $\mu = -2 \wp'(a)$ is uniquely selected for the two solutions (\ref{Lame-sol-1-fin}) and (\ref{Lame-sol-2-fin}) in Theorem \ref{th-main}.

\begin{lemma}
\label{lem-mu}
For the eigenfunctions of the Lax system (\ref{LS}) defined by (\ref{eigenfunction-time}) 
(\ref{Lame-sol-1-fin}) and (\ref{Lame-sol-2-fin}), 
the value of $\mu$ is given by 
\begin{align}
\label{expr-omega}
\mu = -4 \nu^3 k^2 \sn(\gamma) \cn(\gamma) \dn(\gamma), 
\end{align}
\end{lemma}

\begin{proof}
We  note from (\ref{Lame-sol-1-fin}) that $p_1(x_0) = q_1(x_0)$ at 
$$
x_0 := \frac{\alpha - \gamma}{\nu} = -a - \frac{v}{2},
$$
where we have used (\ref{v-alpha-correspondence}) and (\ref{relation-gamma-a}). 
By using the quotient (\ref{rho-basic}) with $\rho(x_0) = 1$ for $x_0 = -a -\frac{v}{2}$, we obtain 
\begin{align*}
\mu &= 4 \zeta^2 \phi(x_0) - b + 2i \zeta \left( 2 \zeta^2 + \phi^2(x_0) - \frac{c}{2} - \phi'(x_0) \right) \\
&= 4 \zeta^2 (\zeta(-a) - \zeta(-a-v) - \zeta(v)) - 2 \wp'(v) 
+ 4i \zeta ( \zeta^2 + \wp(-a) - \wp(v)) \\
&= 4 (\wp(v) - \wp(a)) (\zeta(a+v) - \zeta(a) - \zeta(v)) - 2 \wp'(v), 
\end{align*}
where we have used expressions (\ref{pcW}), (\ref{phinova}), (\ref{philinha}), (\ref{form-3-squared}), and (\ref{Lame-disp-W}). 
By using \cite[8.177]{GR}, 
$$
\zeta(a+v) - \zeta(a) - \zeta(v) = \frac{1}{2} \frac{\wp'(a) - \wp'(v)}{\wp(a) - \wp(a)},
$$
we obtain $\mu = -2 \wp'(a)$. Since $a = (\gamma + i K')/\nu$, we use 
(\ref{rel-2}) and (\ref{rel-Jac-Wei}) to rewrite the expression for $\mu$ in the form:
$$
\mu = -2 \wp'(a) = \frac{4 \nu^3 \cn(\nu a) \dn(\nu a)}{\sn^3(\nu a)} = -4 \nu^3 k^2 \sn(\gamma) \cn(\gamma) \dn(\gamma), 
$$
which yields (\ref{expr-omega}).
\end{proof}

The explicit expression (\ref{expr-omega}) enables a unique parameterization of the square root of the characteristic polynomial $P(\zeta)$ in the definition 
(\ref{character-eq}) of $\mu$. We give explicit expressions 
for $\mu$ in the gaps $(\zeta_2,\zeta_1)$ and $(0,\zeta_3)$ to show that 
the different branches of the square root of $P(\zeta)$ are chosen in between the two gaps. 

\begin{example}
\underline{For $\zeta \in (\zeta_2,\zeta_1)$,} we have $\gamma \in (0,K)$. The expression (\ref{expr-omega}) is equivalent to
\begin{equation}
\label{omega-final-dark}
\mu  = -4 \sqrt{(\zeta_1^2 - \zeta^2) (\zeta^2 - \zeta_2^2) (\zeta^2 - \zeta_3^2)},\quad \zeta \in (\zeta_2,\zeta_1),
\end{equation}
due to (\ref{Lame-disp}). \underline{For $\zeta \in (0,\zeta_3)$,} we write $\gamma = \delta + i K'$ with $\delta \in (2\alpha,K)$. The expression (\ref{expr-omega}) can be rewritten as 
\begin{align}
\label{expr-omega-bright}
\mu = -4 \nu^3 k^2 \sn(\delta + iK') \cn(\delta + iK') \dn(\delta + iK') = \frac{4 \nu^3 \cn(\delta) \dn(\delta)}{\sn^3(\delta)},
\end{align}
which is equivalent to 
\begin{equation}
\label{omega-final}
\mu  = 4 \sqrt{(\zeta_1^2 - \zeta^2) (\zeta_2^2 - \zeta^2) (\zeta_3^2 - \zeta^2)},
\end{equation}
since the dispersion relation (\ref{Lame-disp}) can be rewritten for $\gamma = \delta + i K'$ as
$$
\zeta^2 = \zeta_1^2 - \frac{(\zeta_1^2 - \zeta_3^2)}{\sn^2(\delta)} = 
\zeta_2^2 - \frac{(\zeta_1^2 - \zeta_3^2) \dn^2(\delta)}{\sn^2(\delta)} = 
\zeta_3^2 - \frac{(\zeta_1^2 - \zeta_3^2) \cn^2(\delta)}{\sn^2(\delta)}.
$$
The sign of the square root branch has been changed between (\ref{omega-final-dark}) and  (\ref{omega-final}). When $\zeta \to 0$, we have $\mu \to 4 \zeta_1 \zeta_2 \zeta_3$ from (\ref{omega-final}), in which case the expressions (\ref{eigenfunction-time}) with (\ref{exact-2}) and (\ref{exact-1}) coincide with (2.27)--(2.28) in \cite{AP25}.
\end{example}

\section{Breathers on the elliptic wave background}
\label{sec-4}

We use a linear combination of the two eigenfunctions in Theorem \ref{th-main} in order to construct breathers by using the Darboux transformation (\ref{DT}). This leads to the expression (\ref{DT-expand}). We intend here to show for the gaps $(0,\zeta_3)$ and $(\zeta_2,\zeta_1)$ that the choice (\ref{c1-c2}) produces a bounded and real-valued solution $u = u(x,t)$ of the mKdV equation (\ref{mkdv}). Sections \ref{sec-4-1} and \ref{sec-4-2} report the corresponding results in the two gaps, from which breather solutions in Figures \ref{Fig:Breather2}--\ref{Fig:Breather12} and \ref{Fig:Breather2-dark}--\ref{Fig:Breather12-dark} are constructed. 

\subsection{Breathers in the half gap $(0,\zeta_3)$}
\label{sec-4-1}

By Lemma \ref{lem-beta}, we use the parameterization $\gamma = i K' + \delta$ with $\delta \in (2 \alpha,K)$ and associate the spectral parameter $\zeta$ 
with $\delta$ by using (\ref{zeta-3}). The eigenfunctions (\ref{Lame-sol-1-fin}) and (\ref{Lame-sol-2-fin}) change due to the relations (\ref{Wei-Jac-3}). Both terms in the linear superposition for either (\ref{Lame-sol-1-fin}) or (\ref{Lame-sol-2-fin}) have the same constant multiplicative factor, which we do not write to redefine the eigenfunctions in the form 
\begin{equation}
\label{p1}
\left( \begin{matrix} p_1(x) \\ q_1(x) \end{matrix} \right) = 
\left( \begin{matrix} -i\zeta\frac{\Theta(\nu x + \delta + \alpha)}{\Theta(\nu x + \alpha)} + \zeta_1 \frac{\Theta(0) H(2\alpha + \delta)}{H(\delta) \Theta(2\alpha)} \frac{\Theta(\nu x + \delta - \alpha)}{\Theta(\nu x - \alpha)} \\
-i\zeta\frac{\Theta(\nu x + \delta + \alpha)}{\Theta(\nu x + \alpha)} - \zeta_1 \frac{\Theta(0) H(2\alpha + \delta)}{H(\delta) \Theta(2\alpha)} \frac{\Theta(\nu x + \delta - \alpha)}{\Theta(\nu x - \alpha)} \end{matrix} \right)
e^{-\nu x \frac{H'(\delta)}{H(\delta)}}
\end{equation}
and 
\begin{equation}
\label{p2}
\left( \begin{matrix} p_2(x) \\ q_2(x) \end{matrix} \right) = 
\left( \begin{matrix} \zeta_1 \frac{\Theta(0) H(2\alpha + \delta)}{H(\delta) \Theta(2\alpha)}\frac{\Theta(\nu x - \delta + \alpha)}{\Theta(\nu x + \alpha)} +  i\zeta \frac{\Theta(\nu x - \delta - \alpha)}{\Theta(\nu x - \alpha)} \\
\zeta_1 \frac{\Theta(0) H(2\alpha + \delta)}{H(\delta) \Theta(2\alpha)}\frac{\Theta(\nu x - \delta + \alpha)}{\Theta(\nu x + \alpha)} -  i\zeta \frac{\Theta(\nu x - \delta - \alpha)}{\Theta(\nu x - \alpha)} 
\end{matrix} \right)  e^{\nu x \frac{H'(\delta)}{H(\delta)}}
\end{equation}
All elliptic functions are real-valued for real $\delta \in (2 \alpha,K)$, 
which justify the conditions (\ref{condition-1}) and (\ref{condition-2}). 
The time variable in (\ref{p1}) and (\ref{p2}) is set at $t = 0$. The dependence on time $t$ is given by (\ref{eigenfunction-time}) with $\mu$ given by (\ref{expr-omega-bright}).

\begin{figure}[htb!]
	\includegraphics[width=0.45\textwidth]{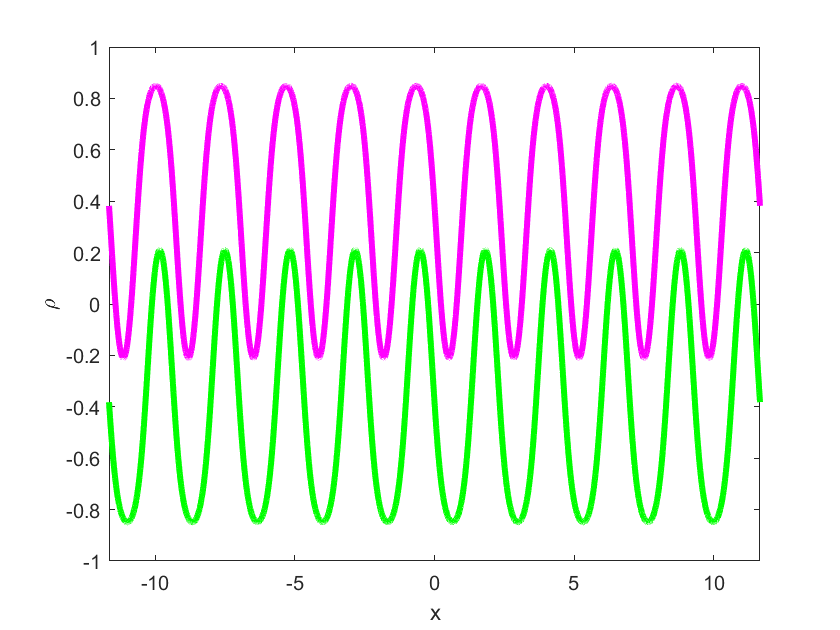}
	\includegraphics[width=0.45\textwidth]{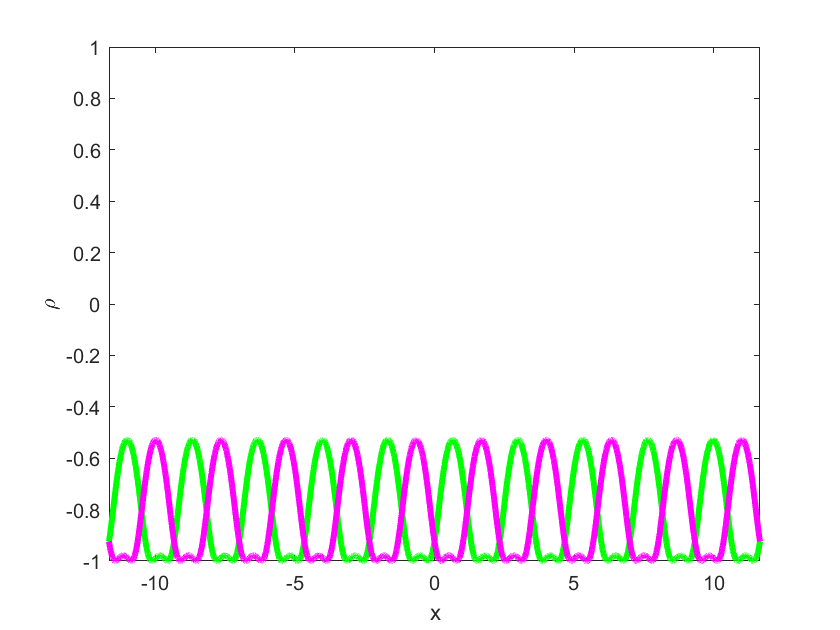}
	\caption{${\rm Re}(\rho)$ versus $x$ (left) and ${\rm Im}(\rho)$ versus $x$ (right), where $\rho = q/p$ is the quotient of the two components of eigenfunctions $\psi = (p,q)^T$ of the spectral problem (\ref{spectral}) with the elliptic potential $\phi$ for $(\zeta_1,\zeta_2,\zeta_3) = (2,1,0.5)$. The spectral parameter $\zeta$ is defined by (\ref{zeta-3}) with $\delta = \frac{1}{2} (2 \alpha + K)$. The green line shows $\rho$ for the first solution (\ref{p1}) and the magenta line shows $\rho$ for the second solution (\ref{p2}).}
	\label{Fig:Rho}
\end{figure}

For the numerical check, we plot in Figure \ref{Fig:Rho} the quotient $\rho = q/p$ (real part is shown on the left panel and imaginary part is on the right panel) for the first solution (\ref{p1}) (green line) and the second solution (\ref{p2}) (magenta line). The quotient $\rho$ is computed in two different ways: by using the elliptic functions (\ref{rho-basic}) with $\mu$ given by (\ref{expr-omega-bright}) and by using the explicit solutions (\ref{p1}) and (\ref{p2}). The two representations are mathematically equivalent and hence the computational difference between them is found within the machine precision error. 

The breather solutions are obtained from (\ref{DT-expand}) with the choice (\ref{c1-c2}), where $\kappa = \frac{H'(\delta)}{H(\delta)}$ is defined by (\ref{kappa}). The breather speed $c_s$ is given by (\ref{speed}) and the spatial decay rate of the breather at infinity is defined by $\kappa$. 
Since $\kappa > 0$ and $\mu > 0$, we have $c_s < c$.

\subsection{Breathers in the gap $(\zeta_2,\zeta_1)$}
\label{sec-4-2}
 
By Lemma \ref{lem-beta}, we use the parameter $\gamma \in (0,K)$ and associate the spectral parameter $\zeta$ by using (\ref{zeta-12}). The eigenfunctions (\ref{Lame-sol-1-fin}) and (\ref{Lame-sol-2-fin}) can be used directly since 
all elliptic functions are real-valued for real $\gamma \in (0,K)$. 
 
For the numerical check, we plot in Figure \ref{Fig:Rho-dark} the quotient $\rho = q/p$ (real part is shown on the left panel and imaginary part is on the right panel) for the first solution (\ref{Lame-sol-1-fin}) (green line) and the second solution (\ref{Lame-sol-2-fin}) (magenta line). The quotient $\rho$ is computed in two different ways: by using the elliptic functions (\ref{rho-basic}) 
with $\mu$ given by (\ref{expr-omega}) and by using the explicit solutions (\ref{Lame-sol-1-fin}) and (\ref{Lame-sol-2-fin}). The computationa difference between the two equivalent representations is again found within the machine precision error.  

\begin{figure}[htb!]
	\includegraphics[width=0.45\textwidth]{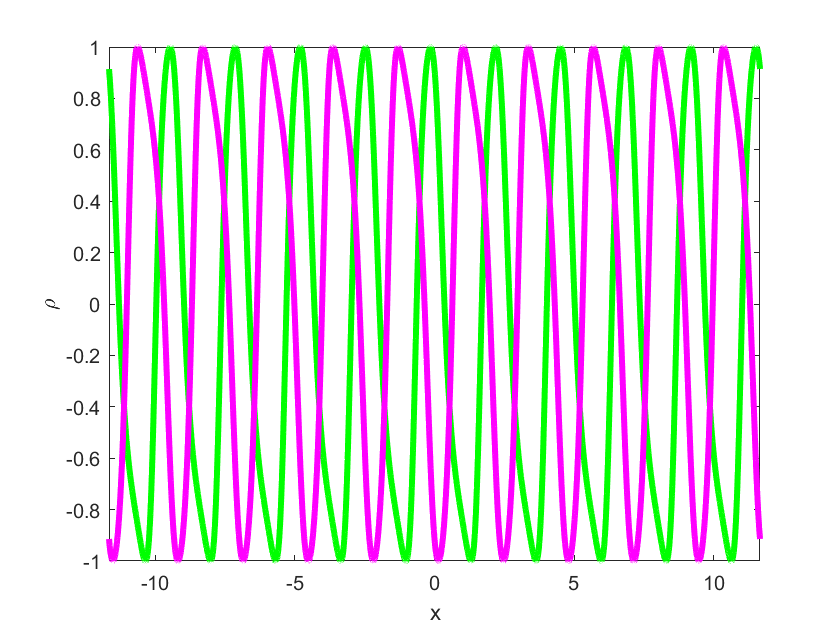}
	\includegraphics[width=0.45\textwidth]{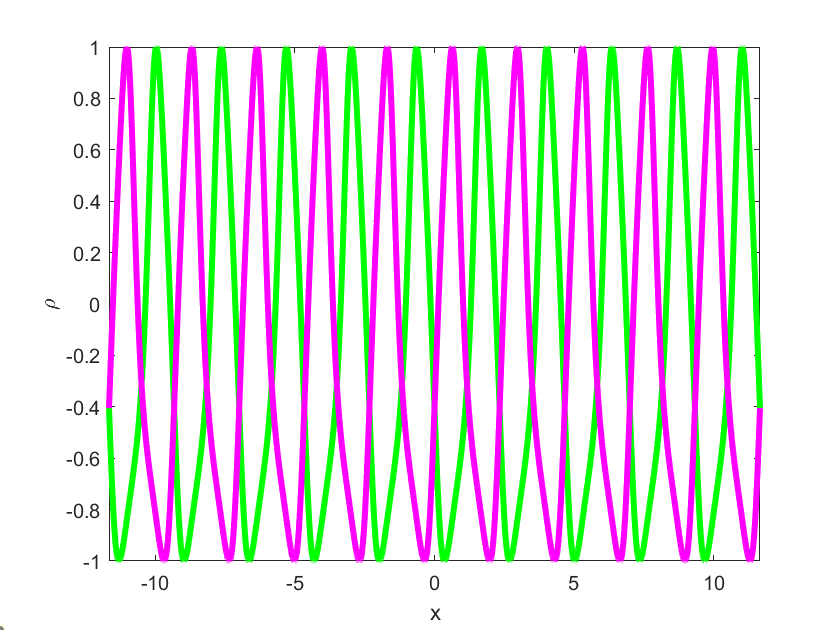}
	\caption{${\rm Re}(\rho)$ versus $x$ (left) and ${\rm Im}(\rho)$ versus $x$ (right), where $\rho = q/p$ is the quotient of the two components of eigenfunctions $\psi = (p,q)^T$ of the spectral problem (\ref{spectral}) with the elliptic potential $\phi$ for $(\zeta_1,\zeta_2,\zeta_3) = (2,1,0.5)$. The spectral parameter $\zeta$ is defined by (\ref{zeta-12}) with $\gamma = \frac{1}{2} K$. The green line shows $\rho$ for the first solution (\ref{Lame-sol-1-fin}) and the magenta line shows $\rho$ for the second solution (\ref{Lame-sol-2-fin}).}
	\label{Fig:Rho-dark}
\end{figure}

If we proceed with the breather solution based on the  transformation formula (\ref{DT-expand}) with eigenfunctions (\ref{Lame-sol-1-fin}) and (\ref{Lame-sol-2-fin}), we get singular solutions for real $x \in \mathbb{R}$. 
However, based on \cite{HMP23,MP24}, we change $x \to x + i K'$ in order to get  bounded breather solutions. In application to the elliptic wave with the profile $\phi$, the transformation $x \to x + i K'$ leaves $\phi$ real-valued but transforms the bounded solution (\ref{form-2}) into the singular solution, 
\begin{equation}
\label{form-2-singular}
\phi(x) = \dfrac{2 (\zeta_1 + \zeta_3)(\zeta_2 + \zeta_3) k^2 \sn^2(\nu x)}{(\zeta_1 + \zeta_3) k^2 {\rm sn}^2(\nu x,k) - (\zeta_1 - \zeta_2)} - \zeta_1 - \zeta_2 - \zeta_3,
\end{equation}
where we have used (\ref{rel-2}). In application to the eigenfunctions (\ref{Lame-sol-1-fin}) and (\ref{Lame-sol-2-fin}), the transformation $x \to x + i K'$ yields up to the constant multiplicative factor:
\begin{equation}
\label{Lame-sol-1-dark}
\left( \begin{matrix} p_1(x) \\ q_1(x) \end{matrix} \right) = 
\left( \begin{matrix}  -i\zeta\frac{\Theta(\nu x + \gamma + \alpha)}{H(\nu x + \alpha)} + \zeta_1 \frac{\Theta(0) \Theta(2\alpha + \gamma)}{\Theta(\gamma) \Theta(2\alpha)} \frac{\Theta(\nu x + \gamma - \alpha)}{H(\nu x - \alpha)} \\
-i\zeta\frac{\Theta(\nu x + \gamma + \alpha)}{H(\nu x + \alpha)} - \zeta_1 \frac{\Theta(0) \Theta(2\alpha + \gamma)}{\Theta(\gamma) \Theta(2\alpha)} \frac{\Theta(\nu x + \gamma - \alpha)}{H(\nu x - \alpha)} \end{matrix} \right) 
e^{-\nu x Z(\gamma)}
\end{equation}
and
\begin{equation}
\label{Lame-sol-2-dark}
\left( \begin{matrix} p_2(x) \\ q_2(x) \end{matrix} \right) = 
\left( \begin{matrix}  \zeta_1 \frac{\Theta(0) \Theta(2\alpha + \gamma)}{\Theta(\gamma) \Theta(2\alpha)}\frac{\Theta(\nu x - \gamma + \alpha)}{H(\nu x + \alpha)} +  i\zeta \frac{\Theta(\nu x - \gamma - \alpha)}{H(\nu x - \alpha)} \\ 
\zeta_1 \frac{\Theta(0) \Theta(2\alpha + \gamma)}{\Theta(\gamma) \Theta(2\alpha)}\frac{\Theta(\nu x - \gamma + \alpha)}{H(\nu x + \alpha)} -  i\zeta \frac{\Theta(\nu x - \gamma - \alpha)}{H(\nu x - \alpha)} \end{matrix} \right) e^{\nu x Z(\gamma)},
\end{equation}
where we have used (\ref{Wei-Jac-3}). All elliptic functions 
in (\ref{Lame-sol-1-dark}) and (\ref{Lame-sol-2-dark}) 
are real-valued for real $\gamma \in (0,K)$, 
which justify the conditions (\ref{condition-1}) and (\ref{condition-2}). 
The time variable in (\ref{Lame-sol-1-dark}) and (\ref{Lame-sol-2-dark}) is set at $t = 0$. The dependence on time $t$ is given by (\ref{eigenfunction-time}) with $\mu$ given by (\ref{omega-final-dark}). 

The breather solutions are obtained from (\ref{DT-expand}) with the choice (\ref{c1-c2}), where $\kappa = Z(\gamma)$ is defined by (\ref{kappa}). Although $\phi$ given by (\ref{form-2-singular}) is singular for real $x \in \mathbb{R}$, the solution $u$ in (\ref{DT-expand}) with the choice (\ref{c1-c2}) is bounded for real $x \in \mathbb{R}$. The breather speed $c_s$ is given by (\ref{speed}) and the spatial decay rate of the breather at infinity is defined by $\kappa$. Since $\kappa > 0$ and $\mu < 0$, we have $c_s > c$.

\section{Factorization of elliptic eigenfunctions}
\label{sec-5}

We have obtained eigenfunctions of the spectral problem (\ref{spectral}) in the explicit form (\ref{Lame-sol-1-fin}) and (\ref{Lame-sol-2-fin}). The spectral parameter $\zeta \in \mathbb{R}$ is represented by the shift parameter $\gamma \in [0,K] \times [0,iK']$ from the dispersion relation (\ref{dispersion}), which also parameterizes the eigenfunctions (\ref{Lame-sol-1-fin}) and (\ref{Lame-sol-2-fin}). This explicit representation of eigenfunctions is now compared with the factorized form discussed in \cite{AP25} and 
used in other works \cite{Geng,Geng3}. We obtain the factorized form
in Section \ref{sec-5-0}. Furthermore, we show that the poles and zeros of the elliptic eigenfunctions depend on the spectral parameter via the branch point singularities, which cannot be unfolded by using elliptic functions. This is done 
for the simpler cases of rational and hyperbolic degenerations of the elliptic functions in Sections \ref{sec-5-1} and \ref{sec-5-2} respectively. 

\subsection{Derivation of the factorized forms}
\label{sec-5-0}

By using the quotient (\ref{rho-basic}) with $b = 4 \zeta_1 \zeta_2 \zeta_3$ and 
$c = 2 (\zeta_1^2 + \zeta_2^2 + \zeta_3^2)$, and a unique choice for $\mu$ in (\ref{expr-omega}), we introduce zeros of numerators and denominators 
in the general case $0 < \zeta_3 < \zeta_2 < \zeta_1$. In particular, we prove in Lemma \ref{prop-zeros} below 
that there exist exactly two symmetric  pairs of four roots in $x$ labeled as 
$\{ \pm x_1, \pm x_2\}$ in $[-\omega,\omega] \times [-\omega',\omega']$ 
for the numerator of the first quotient in (\ref{rho-basic}),
\begin{equation}
\label{square-root-1}
2 i \zeta \left[ \phi^2(x) + 2 \zeta^2 - \zeta_1^2 - \zeta_2^2 - \zeta_3^2 \right] + \mu = 0,
\end{equation}
and exactly two symmetric  pairs of four roots in $x$ labeled as 
$\{ \pm x_1^*, \pm x_2^*\}$ in $[-\omega,\omega] \times [-\omega',\omega']$ 
for the denominator of the second quotient in (\ref{rho-basic}),
\begin{equation}
\label{square-root-2}
2 i \zeta \left[ \phi^2(x) + 2 \zeta^2 - \zeta_1^2 - \zeta_2^2 - \zeta_3^2 \right] - \mu = 0,
\end{equation}
where $\omega$ and $\omega'$ are given by (\ref{periods}).
Furthermore, we prove in Lemma \ref{prop-zeros-2} that two roots from $\{ \pm x_1, \pm x_2\}$ and two roots from $\{ \pm x_1^*, \pm x_2^*\}$ are 
four roots in $x$ for the denominator of the first quotient in (\ref{rho-basic}),
\begin{equation}
\label{zero-denominator}
\zeta^2 \phi(x) - \frac{i}{2} \zeta \phi'(x) - \zeta_1 \zeta_2 \zeta_3 = 0,
\end{equation}
in $[-\omega,\omega] \times [-\omega',\omega']$, whereas the complement from the two sets of roots give four roots in $x$ for the numerator of the second quotient in (\ref{rho-basic}),
\begin{equation}
\label{zero-denominator-recipr}
\zeta^2 \phi(x) + \frac{i}{2} \zeta \phi'(x) - \zeta_1 \zeta_2 \zeta_3 = 0.
\end{equation}
For convenience, we denote roots of (\ref{square-root-1}) and (\ref{square-root-2}), which are simultaneously roots of (\ref{zero-denominator}) by $\{x_1,x_2 \}$ and $\{ x_1^*,x_2^*\}$ respectively. Since $\phi$ is even and $\phi'$ is odd in $x$, the symmetry between (\ref{zero-denominator}) and (\ref{zero-denominator-recipr}) implies that 
roots of (\ref{zero-denominator-recipr}) are given by $\{-x_1,-x_2 \}$ and $\{ -x_1^*,-x_2^*\}$.

\begin{lemma}
	\label{prop-zeros}
There exist exactly two symmetric pairs of four roots in $[-\omega,\omega] \times [-\omega',\omega']$ with respect to $x$ for either  (\ref{square-root-1}) or (\ref{square-root-2}). 
\end{lemma}

\begin{proof}
	For simplicity, we deal with solutions of (\ref{square-root-1}). By using the analytical representation (\ref{form-3-squared}) with (\ref{pcW}), we can rewrite (\ref{square-root-1}) in the equivalent form 
	$$
	2 i \zeta \left[ \wp\left( x + \frac{v}{2} \right) + 
	\wp\left( x - \frac{v}{2} \right) + 2 \zeta^2 - \frac{2}{3} (\zeta_1^2 + \zeta_2^2 + \zeta_3^2) \right] + \mu = 0.
	$$
By using (\ref{rel-Jac-Wei}), (\ref{rel-par-e-zeta}), and (\ref{v-alpha-correspondence}), we obtain 
	\begin{equation}
	\label{roots-sn-squared}
	2 i \zeta (\zeta_1^2 - \zeta_2^2) \left[  \sn^2(z + \alpha) + \sn^2(z - \alpha) \right]  +  4i \zeta (\zeta^2 - \zeta_1^2) + \mu = 0,
	\end{equation}
	where $z = \nu x$. With the help of the addition formula
	\begin{equation*}
	\sn(u \pm v) = \frac{\sn(u) \cn(v) \dn(v) \pm \sn(v) \cn(u) \dn(u)}{1 - k^2 \sn^2(u) \sn^2(v)},
	\end{equation*}
	we rewrite (\ref{roots-sn-squared}) in the form 
	$$
	4i \zeta (\zeta_1^2 - \zeta_2^2) \frac{\sn^2(z) \cn^2(\alpha) \dn^2(\alpha) + \sn^2(\alpha) \cn^2(z) \dn^2(z)}{[1 - k^2 \sn^2(z) \sn^2(\alpha)]^2} + 
	4i \zeta (\zeta^2 - \zeta_1^2) + \mu = 0.
	$$
Eliminating $\sn(\alpha)$, $\cn(\alpha)$, and $\dn(\alpha)$ from (\ref{elliptic-alpha}) yields
	\begin{align}
4i \zeta (\zeta_1 - \zeta_2) (\zeta_1 + \zeta_3) \frac{(\zeta_2 + \zeta_3)^2 \sn^2(z) + (\zeta_1^2 - \zeta_3^2) \cn^2(z) \dn^2(z)}{[(\zeta_1 + \zeta_3) - (\zeta_1 - \zeta_2) \sn^2(z)]^2} + 4i \zeta (\zeta^2 - \zeta_1^2) + \mu = 0.
	\label{root-elliptic}
	\end{align}
By using the fundamental relations (\ref{fund-rel}), equation (\ref{root-elliptic}) yields a bi-quadratic equation for $\sn^2(z)$ with two roots. For each root of $\sn^2(z)$, there exist exactly two solutions of $\sn^2(z) = w \in \mathbb{C}$ in $[-K,K] \times [-iK',iK']$, see Proposition 3.1 in \cite{AP25}. This construction yields exactly four solutions of (\ref{square-root-1}) for $x$ in $[-\omega,\omega] \times [-\omega',\omega']$. Since $\sn^2(z)$ is even in $z$, the four roots of (\ref{square-root-1}) form two symmetric pairs. The proof for (\ref{square-root-2}) is identical.
\end{proof}

\begin{lemma}
	\label{prop-zeros-2} 
	Two roots of (\ref{square-root-1}) and two roots of (\ref{square-root-2}) are exactly four roots of (\ref{zero-denominator}) in $[-\omega,\omega] \times [-\omega',\omega']$ with respect to $x$.
\end{lemma}

\begin{proof}
By squaring (\ref{zero-denominator}) and using (\ref{third}), we obtain 
\begin{align*}
& \quad -\zeta^2 (\phi')^2 = 4 (\zeta^2 \phi - \zeta_1 \zeta_2 \zeta_3)^2 \\
&= -\zeta^2 (\phi^4 -2 (\zeta_1^2 + \zeta_2^2 + \zeta_3^2) \phi^2 + 8 \zeta_1 \zeta_2 \zeta_3 \phi + \zeta_1^4 + \zeta_2^4 + \zeta_3^4 - 2(\zeta_1^2 \zeta_2^2 + \zeta_1^2 \zeta_3^2 + \zeta_2^2 \zeta_3^2)).
\end{align*}
This yields a bi-quadratic equation 
\begin{equation}
\label{square-root}
4 \zeta^2 \left( \phi^2 + 2 \zeta^2 - \zeta_1^2 - \zeta_2^2 - \zeta_3^2 \right)^2 + \mu^2 = 0,
\end{equation}
where we have used (\ref{character-eq}) with $P(\zeta) = (\zeta^2 - \zeta_1^2)(\zeta^2 - \zeta_2^2)(\zeta^2 - \zeta_3^2)$. Extracting the square roots of (\ref{square-root}) yields (\ref{square-root-1}) and (\ref{square-root-2}). 
Since $\phi$ is even and $\phi'$ is odd in $x$, the roots of (\ref{zero-denominator}) do not form symmetric pairs. This implies that 
two of the four roots of (\ref{zero-denominator}) are given by roots 
of (\ref{square-root-1}) and the other two roots of (\ref{zero-denominator}) are given by roots of (\ref{square-root-2}). 
\end{proof}

With the use of Lemmas \ref{prop-zeros} and \ref{prop-zeros-2}, we obtain the factorized form for the eigenfunctions of the spectral problem (\ref{spectral-problem}). The factorized form provides an alternative 
definition of the eigenfunctions compared to the representations 
(\ref{Lame-sol-1-fin}) and (\ref{Lame-sol-2-fin}) in Theorem \ref{th-main}. 
The corresponding result is given by the following theorem. 

\begin{theorem}
	\label{th-factorizaton}
	Let $u(x,t) = \phi(x+ct)$ be defined by (\ref{form-2}) for $0 < \zeta_3 < \zeta_2 < \zeta_1$. Define $\{ x_1,x_2\}  \in [-\omega,\omega] \times [-\omega',\omega']$ and $\{ x_1^*,x_2^*\}  \in [-\omega,\omega] \times [-\omega',\omega']$ by roots of (\ref{square-root-1}) and (\ref{square-root-2}) respectively, which are simultaneously roots of (\ref{zero-denominator}). There exist two linearly independent solutions of the Lax system (\ref{LS}) in the form (\ref{eigenfunction-time}) with 
	\begin{equation}
	\label{p-gen}
	\left( \begin{matrix} p_1(x)  \\ q_1(x) \end{matrix} \right) = 
	\left( \begin{matrix} 
	\frac{H(\nu (x-x_1)) H(\nu (x-x_2))}{\Theta(\nu x-\alpha) \Theta(\nu x+ \alpha) \Theta(\alpha - \nu x_1) \Theta(\alpha - \nu x_2)} e^{-\frac{i\pi \nu}{2K} (x_1 + x_2)} \\ -
	\frac{H(\nu (x+x_1^*)) H(\nu (x+x_2^*))}{\Theta(\nu x-\alpha) \Theta(\nu x+ \alpha) \Theta(\alpha + \nu x_1^*) \Theta(\alpha + \nu x_2^*)} e^{\frac{i\pi \nu}{2K} (x_1^* + x_2^*)} \end{matrix} \right) e^{s x}.
	\end{equation}
	and
	\begin{equation}
	\label{p-gen-1}
	\left( \begin{matrix} p_2(x) \\ q_2(x) \end{matrix} \right) = 
	\left( \begin{matrix}
	\frac{H(\nu (x-x_1^*)) H(\nu (x-x_2^*))}{\Theta(\nu x-\alpha) \Theta(\nu x+ \alpha) \Theta(\alpha + \nu x_1^*) \Theta(\alpha + \nu x_2^*)} e^{\frac{i\pi \nu}{2K} (x_1^* + x_2^*)} \\
	\frac{H(\nu (x+x_1)) H(\nu (x+x_2))}{\Theta(\nu x-\alpha) \Theta(\nu x+ \alpha) \Theta(\alpha - \nu x_1) \Theta(\alpha - \nu x_2)} e^{-\frac{i\pi \nu }{2K} (x_1 + x_2)} \end{matrix} \right) e^{-s x},
	\end{equation}
	where
\begin{align}
s = i \frac{4 \zeta^4 - 4 \zeta^2 (\zeta_2 + \zeta_3 - \zeta_1)^2 
	- i \zeta \mu + 4 \zeta_1 \zeta_2 \zeta_3 (\zeta_2 + \zeta_3 - \zeta_1)}{4 \zeta^3 +4  \zeta (\zeta_2 \zeta_3 - \zeta_1 \zeta_2 - \zeta_1 \zeta_3) - i \mu} 
+ \frac{\nu H'(\nu x_1)}{H(\nu x_1)} + \frac{\nu H'(\nu x_2)}{H(\nu x_2)}.
\label{s-final}
\end{align}	
\end{theorem}

\begin{proof}
	The proof is computational and consists of the following steps. 

\underline{The first quotient $\rho = \frac{q_1}{p_1}$.} \\
Poles of both numerator and denominator in (\ref{rho-basic}) coincide with the poles of $\phi^2(x)$ and $\phi'(x)$, which are double poles at the same locations $\pm \frac{v}{2}$, see (\ref{philinha}) and (\ref{form-3-squared}). Since both the numerator and the denominator in each quotient of (\ref{rho-basic}) are elliptic functions, which have four zeros in $[-\omega,\omega] \times [-\omega',\omega']$  and two double poles at $\pm \frac{v}{2}$, the factorized form for elliptic functions, see Proposition 3.3 in \cite{AP25}, implies 
for both quotients in (\ref{rho-basic}) that 
\begin{align}
\rho(x) &= C 
\frac{H(\nu (x + x_1^*)) H(\nu (x + x_2^*)) 
	\cancel{H(\nu (x - x_1^*)) H(\nu (x - x_2^*))}}{H(\nu (x - x_1)) H(\nu (x - x_2)) 
	\cancel{H(\nu (x - x_1^*)) H(\nu (x - x_2^*))}} \notag \\
&= C \frac{H(\nu (x + x_1^*)) H(\nu (x + x_2^*)) 
	\cancel{H(\nu (x + x_1)) H(\nu (x + x_2))}}{H(\nu (x - x_1)) H(\nu (x - x_2)) 
	\cancel{H(\nu (x + x_1)) H(\nu (x + x_2))}}, \label{rho-first}
\end{align} 
for some constant $C \in \R$. To obtain $C$ explicitly, we consider the quotients in (\ref{rho-basic}) at the pole singularity as $x \to -\frac{v}{2}$, that is, 
$x \to \nu^{-1} (iK'+\alpha)$.  By using (\ref{phi-asymp}) and (\ref{rho-first}), we obtain 
\begin{equation}
\label{constant-C-unique}
\lim_{x \to \nu^{-1} (iK'+\alpha)} \rho(x) = -1 = C e^{-\frac{i\pi \nu}{2K} (x_1 + x_2 + x_1^* + x_2^*)} \frac{\Theta(\alpha + \nu x_1^*) \Theta(\alpha + \nu x_2^*)}{\Theta(\alpha - \nu x_1) \Theta(\alpha - \nu x_2)}.
\end{equation}
This yields 
\begin{align}
\rho(x) &= -e^{\frac{i\pi \nu}{2K} (x_1 + x_2 + x_1^* + x_2^*)} 
\frac{H(\nu (x + x_1^*)) H(\nu (x + x_2^*)) \Theta(\alpha - \nu x_1) \Theta(\alpha - \nu x_2)}{H(\nu (x - x_1)) H(\nu (x - x_2)) \Theta(\alpha + \nu x_1^*) \Theta(\alpha + \nu x_2^*)}, \label{rho-final}
\end{align} 
which is the quotient $\rho = \frac{q_1}{p_1}$ of the first solution $\psi = (p_1,q_1)^T$ of the spectral problem (\ref{spectral}). The two poles and two zeros of the elliptic function $\rho$ in (\ref{rho-first}) are related by 
\begin{equation*}
x_1 + x_2 + x_1^* + x_2^* = 0 \; \mbox{\rm mod } (2\omega, 2 \omega').
\end{equation*}

\underline{The second quotient $\rho = \frac{q_2}{p_2}$.} \\
We change $\mu$ to $-\mu$ in (\ref{rho-basic}) and use the same notations for the sets of roots in Lemmas \ref{prop-zeros} and \ref{prop-zeros-2}. This yields
\begin{align}
\rho(x) &= C
\frac{H(\nu (x+x_1)) H(\nu (x+x_2)) 
	\cancel{H(\nu (x-x_1)) H(\nu (x-x_2))}}{H(\nu (x-x_1^*)) H(\nu (x-x_2^*)) 
	\cancel{H(\nu (x-x_1)) H(\nu (x-x_2))}} \notag \\
&= C \frac{H(\nu (x+x_1)) H(\nu (x+x_2)) 
\cancel{H(\nu (x+x_1^*)) H(\nu (x+x_2^*))}}{H(\nu (x-x_1^*)) H(\nu (x-x_2^*)) 
\cancel{H(\nu (x+x_1^*)) H(\nu (x+x_2^*)}}, \label{rho-second}
\end{align} 
for some constant $C \in \R$. To obtain $C$ explicitly, we consider the singular behavior as $x \to \frac{v}{2}$, that is, $x \to -\nu^{-1} (iK'+\alpha)$, for which 
\begin{equation}
\label{constant-C-unique-second}
\lim_{x \to -\nu^{-1} (iK'+\alpha)} \rho(x) = 1 = C e^{\frac{i\pi \nu}{2K}  (x_1 + x_2 + x_1^* + x_2^*)} \frac{\Theta(\alpha - \nu x_1) \Theta(\alpha - \nu x_2)}{\Theta(\alpha + \nu x_1^*) \Theta(\alpha + \nu x_2^*)}.
\end{equation}
This yields 
\begin{align}
\rho(x) &= e^{-\frac{i\pi \nu }{2K}  (x_1 + x_2 + x_1^* + x_2^*)} 
\frac{H(\nu (x+ x_1)) H(\nu (x+ x_2))\Theta(\alpha + \nu x_1^*) \Theta(\alpha + \nu x_2^*)}{H(\nu (x - x_1^*)) H(\nu (x - x_2^*)) \Theta(\alpha - \nu x_1) \Theta(\alpha - \nu x_2)}, 
\label{rho-final-second}
\end{align} 
which is the quotient $\rho = \frac{q_2}{p_2}$ of the second solution $\psi = (p_2,q_2)^T$ of the spectral problem (\ref{spectral}). The limit (\ref{constant-C-unique-second}) is different from the limit (\ref{constant-C-unique}) and this yields to the asymmetry between (\ref{rho-final}) and (\ref{rho-final-second}) in the sense that the two expressions are not obtained by replacing $\{x_1,x_2\}$ with $\{x_1^*,x_2^*\}$ and vice versa. \\

\underline{Factorized form for $(p_1,q_1)$.} \\
If $\rho$ is defined, the components of $\psi = (p,q)^T$ can be found from the first-order  equations:
\begin{equation}
\label{p-eq}
\partial_x p = (i \zeta + \phi \rho) p, \quad 
\partial_x q = (\phi \rho^{-1} - i \zeta) q.
\end{equation}
By using the second quotient in (\ref{rho-basic}) in the first equation of system (\ref{p-eq}), we write
\begin{equation}
\label{p-integrals}
\partial_x \log p_1 = i \zeta + 4 \frac{\zeta^2 \phi^2 + \frac{i}{2} \zeta \phi \phi' - \zeta_1 \zeta_2 \zeta_3 \phi}{2 i \zeta (\phi^2 + 2 \zeta^2 - \zeta_1^2 - \zeta_2^2 - \zeta_3^2) + \mu}.
\end{equation}
For simplicity of notations, we define
\begin{equation}
\label{upsilon}
\upsilon^2 := \zeta_1^2 + \zeta_2^2 + \zeta_3^2 - 2 \zeta^2 - \frac{\mu}{2i \zeta}.
\end{equation}
In order to compute integrals explicitly, we recall the linear fractional tranformation between the elliptic profile $\phi$ and the Weierstrass' function $\wp$, see Lemma 3.4 in \cite{AP25}:
\begin{equation}
\label{lin-frac-trans}
\phi(x) = \frac{\alpha_1 \wp(x) + \beta_1}{\gamma_1 \wp(x) + \delta_1},
\end{equation}
where 
\begin{align*}
\alpha_1 &= \zeta_2 + \zeta_3 - \zeta_1,\\
\beta_1 &= \frac{1}{3} \zeta_1 (\zeta_1^2 - 2 \zeta_2^2 - 2 \zeta_3^2 - 3 \zeta_2 \zeta_3) + \frac{1}{3} (\zeta_2 + \zeta_3) (2 \zeta_1^2 - \zeta_2^2 - \zeta_3^2 + 3 \zeta_2 \zeta_3), \\
\gamma_1 &= 1,\\
\delta_1 &= \zeta_1 (\zeta_2 + \zeta_3) - \zeta_2 \zeta_3 -\frac{1}{3} (\zeta_1^2 + \zeta_2^2 + \zeta_3^2).
\end{align*}
By using (\ref{lin-frac-trans}), 5.141.5 on p. 626 in \cite{GR}, 
and 8.193.2-3 on p. 880 in \cite{GR}, we obtain 
\begin{align*}
\int \frac{dx}{\phi(x) \pm \upsilon} &= \int \frac{\gamma_1 \wp(x) + \delta_1}{(
	\alpha_1 \pm \upsilon \gamma_1) \wp(x) + (\beta_1 \pm \upsilon \delta_1)} dx \\
&= \frac{\gamma_1 x}{\alpha_1 \pm \upsilon \gamma_1} + \frac{\beta_1 \gamma_1 - \alpha_1 \delta_1}{(\alpha_1 \pm \upsilon \gamma_1)^2 \wp'(x_{\pm})} \left[ \log \frac{H(\nu(x + x_{\pm}))}{H(\nu(x-x_{\pm}))} - 2 \nu x \frac{H'(\nu x_{\pm})}{H(\nu x_{\pm})} \right], 
\end{align*}
where $x_{\pm}$ are roots of $\phi(x_{\pm}) \pm \upsilon = 0$ respectively. We note that 
\begin{align*}
\phi'(x_{\pm}) = \frac{\alpha_1 \delta_1 - \beta_1 \gamma_1}{(\gamma_1 \wp(x_{\pm}) + \delta_1)^2} \wp'(x_{\pm}) = \frac{(\alpha_1 \pm \upsilon \gamma_1)^2}{\alpha_1 \delta_1 - \beta_1 \gamma_1} \wp'(x_{\pm}). 
\end{align*}
The choice of the sign in $\pm x_+$ and $\pm x_-$ do not change the outcome of integration. Since $\phi^2(x) = \upsilon^2$ is equivalent to (\ref{square-root-1}), roots of which have been denoted as $\{ \pm x_1, \pm x_2\}$, we can further fix the sign of $x_{\pm}$ to require them to be roots 
of (\ref{zero-denominator}), that is, $x_1 = x_+$ and $x_2 = x_-$. As a result, we have 
\begin{align*}
i \zeta \phi'(x_{\pm}) = 2 (\zeta^2 \phi(x_{\pm}) - \zeta_1 \zeta_2 \zeta_3) =  2(\mp \upsilon \zeta^2 - \zeta_1 \zeta_2 \zeta_3).
\end{align*}
By using these expressions, we compute the integrals in (\ref{p-integrals}) as follows:
\begin{align*}
\log p_1 &= i \zeta x + \int \frac{\phi \phi' }{\phi^2 - \upsilon^2} dx
+ \frac{2i \zeta_1 \zeta_2 \zeta_3}{\zeta} \int \frac{\phi}{\phi^2 - \upsilon^2} dx 
- 2 i \zeta \int \frac{\phi^2}{\phi^2 - \upsilon^2} dx \\
&= -i \zeta x + \frac{1}{2} \log(\phi^2 - \upsilon^2) \\
& \qquad 
+ \frac{i \zeta_1 \zeta_2 \zeta_3}{\zeta} \int \left[ \frac{1}{\phi + \upsilon} + \frac{1}{\phi - \upsilon} \right] dx 
+ i \zeta \upsilon  \int \left[ \frac{1}{\phi + \upsilon} - \frac{1}{\phi - \upsilon} \right] dx \\
&= -i \zeta x + \frac{1}{2} \log(\phi^2 - a^2)
+ \frac{1}{2} \phi'(x_+) \int \frac{dx}{\phi + a}  
+ \frac{1}{2} \phi'(x_-)  \int \frac{dx}{\phi - a} \\
&= -i \zeta x + \frac{1}{2} \log(\phi^2 - \upsilon^2)
+ \frac{\gamma_1 x}{2} \left[ \frac{\phi'(x_+)}{\alpha_1 + \upsilon \gamma_1} 
+ \frac{\phi'(x_-)}{\alpha_1 - \upsilon \gamma_1} \right] \\
& \qquad -\frac{1}{2} \log \frac{H(\nu(x+x_1)) H(\nu (x+x_2))}{H(\nu (x-x_1)) H(\nu (x-x_2))} + \nu x \left[ \frac{H'(\nu x_1)}{H(\nu x_1)} + \frac{H'(\nu x_2)}{H(\nu x_2)} \right].
\end{align*}
Since $\phi^2(x) - \upsilon^2$ is an elliptic function with roots given by roots of (\ref{square-root-1}), we use Lemma \ref{prop-zeros} and write it in the factorized form 
\begin{equation}
\label{phi-squared-1}
\phi^2(x) - \upsilon^2 = C \frac{H(\nu (x-x_1)) H(\nu (x-x_2)) H(\nu (x+x_1)) H(\nu (x+x_2))}{\Theta^2(\nu x - \alpha) \Theta^2(\nu x + \alpha)},
\end{equation}
where $C$ is a constant. 
By exponentiating, we obtain the explicit formula for $p$ 
\begin{align}
p_1 &= C_1 e^{s x} \frac{\sqrt{H(\nu (x-x_1)) H(\nu (x-x_2)) 	\cancel{H(\nu (x+x_1))} 	\cancel{H(\nu (x+x_2))}}}{\Theta(\nu x - \alpha) \Theta(\nu x + \alpha)} \notag \\
& \qquad \qquad \qquad \qquad
\times \frac{\sqrt{H(\nu (x-x_1)) H(\nu (x-x_2))}}{\sqrt{	\cancel{H(\nu (x+x_1))} 	\cancel{H(\nu (x+x_2))}}} \notag \\
&= C_1 e^{s x}  \frac{H(\nu (x-x_1)) H(\nu (x-x_2))}{\Theta(\nu x - \alpha) \Theta(\nu x + \alpha)}, \label{p-1}
\end{align}
where $C_1$ is arbitrary constant and 
\begin{align*}
s := -i\zeta + \frac{2 \gamma_1 (\upsilon^2 \gamma_1 \zeta^2 - \alpha_1 \zeta_1 \zeta_2 \zeta_3)}{i \zeta (\alpha_1^2 - \gamma_1^2 \upsilon^2)} + \frac{\nu H'(\nu x_1)}{H(\nu x_1)} + \frac{\nu H'(\nu x_2)}{H(\nu x_2)},
\end{align*}
which yields (\ref{s-final}) from (\ref{upsilon}) and (\ref{lin-frac-trans}).
By using the quotient (\ref{rho-final}) for $\rho$, we hence obtain the explicit formula for $q_1 = p_1 \rho$ as 
\begin{align}
q_1 = C_1 C e^{s x}  \frac{H(\nu (x+x_1^*)) H(\nu (x+x_2^*))}{\Theta(\nu x - \alpha) \Theta(\nu x + \alpha)}, \label{q-1}
\end{align}
where the constants $C$ is defined uniquely in (\ref{constant-C-unique}). 
By choosing 
$$
C_1 = \frac{e^{-\frac{i\pi \nu}{2K} (x_1 + x_2)}}{\Theta(\alpha - \nu x_1) \Theta(\alpha - \nu x_2)} , 
$$
we obtain the explicit expressions for the first solution in the form (\ref{p-gen}).

\underline{Factorized form for $(p_2,q_2)$.} \\
We replace $\mu$ by $-\mu$ in (\ref{p-integrals}). This replaces $\upsilon^2$ 
in (\ref{upsilon}) with 
$$
(\upsilon^*)^2 := \zeta_1^2 + \zeta_2^2 + \zeta_3^2 - 2 \zeta^2 + \frac{\mu}{2i \zeta}. 
$$
Since $\phi^2(x) - (\upsilon^*)^2$ is an elliptic function with roots given by roots of (\ref{square-root-2}), we use Lemma \ref{prop-zeros} and write it in the factorized form
\begin{equation}
\label{phi-squared-2}
\phi^2(x) - (\upsilon^*)^2 = C^* \frac{H(\nu (x-x_1^*)) H(\nu (x-x_2^*)) 
	H(\nu (x+x_1^*)) H(\nu (x+x_2^*))}{\Theta^2(\nu x - \alpha) \Theta^2(\nu x + \alpha)},
\end{equation}
where $C^*$ is another constant. Computations remain the same with $\{\pm x_1,\pm x_2\}$ replaced by $\{ \pm x_1^*, \pm x_2^*\}$ so that roots of $\phi(x_{\pm}^*) \pm \upsilon^* = 0$ are placed in the correspondence with $x_1^* = x_+^*$ and $x_2^* = x_-^*$. As result of similar computations, we obtain 
\begin{align}
 \label{p-2}
p_2 = C_2 e^{s^* x}  \frac{H(\nu (x-x_1^*)) H(\nu (x-x_2^*))}{\Theta(\nu x - \alpha) \Theta(\nu x + \alpha)}
\end{align}
and 
\begin{align}
 \label{q-2}
q_2 = C_2 C e^{s^* x}  \frac{H(\nu (x+x_1)) H(\nu (x+x_2))}{\Theta(\nu x - \alpha) \Theta(\nu x + \alpha)},
\end{align}
where $C_2$ is arbitrary constant, the constant $C$ is defined uniquely in (\ref{constant-C-unique-second}), and 
\begin{align}
s^* := i \frac{4 \zeta^4 - 4 \zeta^2 (\zeta_2 + \zeta_3 - \zeta_1)^2 
	+ i \zeta \mu + 4 \zeta_1 \zeta_2 \zeta_3 (\zeta_2 + \zeta_3 - \zeta_1)}{4 \zeta^3 + 4 \zeta (\zeta_2 \zeta_3 - \zeta_1 \zeta_2 - \zeta_1 \zeta_3) + i \mu} 
+ \frac{\nu H'(\nu x_1^*)}{H(\nu x_1^*)} + \frac{\nu H'(\nu x_2^*)}{H(\nu x_2^*)}.
\label{s-star}
\end{align}
We prove in Lemma \ref{lem-s} that $s^* = -s$. Hence, by choosing 
$$
C_2 = \frac{e^{\frac{i\pi \nu}{2K} (x_1^* + x_2^*)}}{\Theta(\alpha + \nu x_1^*) \Theta(\alpha + \nu x_2^*)}, 
$$
we obtain the explicit expressions for the second solution in the form (\ref{p-gen-1}).
\end{proof}

It remains to prove that $s^* = -s$, which we do in the following lemma.

\begin{lemma}
	\label{lem-s}
It follows that $s + s^* = 0$, where $s$ and $s^*$ are given by (\ref{s-final}) and (\ref{s-star}).
\end{lemma}

\begin{proof}
	We recall that the Wronskian determinant of the two solutions of the linear system (\ref{LS}) is independent of $(x,t)$. By using (\ref{constant-C-unique}), (\ref{constant-C-unique-second}), (\ref{p-1}), (\ref{q-1}), (\ref{p-2}), and (\ref{q-2}), we obtain for the Wronskian $W$ of the two solutions:
\begin{align*}
W &= e^{(s+s^*)x} \frac{H(\nu (x-x_1)) H(\nu (x-x_2)) H(\nu (x+x_1)) 
	H(\nu (x+x_2))}{\Theta^2(\nu x - \alpha) \Theta^2(\nu x + \alpha) 
\Theta^2(\alpha - \nu x_1) \Theta^2(\alpha - \nu x_2)} e^{-\frac{i \pi \nu }{K} (x_1 + x_2)} \\
& \quad + e^{(s+s^*)x} \frac{H(\nu (x - x_1^*)) H(\nu (x - x_2^*)) H(\nu (x + x_1^*)) H(\nu (x + x_2^*))}{\Theta^2(\nu x - \alpha) \Theta^2(\nu x + \alpha) 
	\Theta^2(\alpha + \nu x_1^*) \Theta^2(\alpha + \nu x_2^*)} e^{\frac{i \pi \nu}{K} (x_1^* + x_2^*)}.
\end{align*}
By using (\ref{phi-squared-1}) and (\ref{phi-squared-2}), this can be rewritten in the form 
\begin{align*}
W &= e^{(s+s^*)x} \frac{\phi^2(x) - \upsilon^2}{C
	\Theta^2(\alpha - \nu x_1) \Theta^2(\alpha - \nu x_2)} e^{-\frac{i \pi \nu}{K} (x_1 + x_2)} \\
& \quad + e^{(s+s^*)x} \frac{\phi^2(x) - (\upsilon^*)^2}{C^*
	\Theta^2(\alpha + \nu x_1^*) \Theta^2(\alpha + \nu x_2^*)} e^{\frac{i \pi \nu}{K} (x_1^* + x_2^*)}.
\end{align*}
We show that the coefficient in front of $e^{(s+s^*)x} \phi^2(x)$ vanishes identically. Indeed, by using (\ref{Wei-Jac-3}), we compute from (\ref{phi-squared-1}) and (\ref{phi-squared-2}) as $\nu x \to iK' + \alpha$ that 
\begin{align*}
(\zeta_1 - \zeta_2 - \zeta_3)^2 \frac{\Theta^4(\alpha)}{H^4(\beta)} 
\Theta^2(\alpha - \beta) \Theta^2(\alpha + \beta) &= C \Theta(\alpha - \nu x_1) \Theta(\alpha - \nu x_2) \Theta(\alpha + \nu x_1) \Theta(\alpha + \nu x_2), \\
(\zeta_1 - \zeta_2 - \zeta_3)^2 \frac{\Theta^4(\alpha)}{H^4(\beta)} 
\Theta^2(\alpha - \beta) \Theta^2(\alpha + \beta) &= C^* \Theta(\alpha - \nu x_1^*) \Theta(\alpha - \nu x_2^*) \Theta(\alpha + \nu x_1^*) \Theta(\alpha + \nu x_2^*).
\end{align*}
The coefficient in front of $e^{(s+s^*)x} \phi^2(x)$ in $W$ is zero if and only if 
$$
\frac{\Theta(\alpha + \nu x_1) \Theta(\alpha + \nu x_2)}{\Theta(\alpha - \nu x_1) \Theta(\alpha - \nu x_2)} e^{-\frac{i \pi}{K} \nu (x_1 + x_2)} + \frac{\Theta(\alpha - \nu x_1^*) \Theta(\alpha - \nu x_2^*)}{\Theta(\alpha + \nu x_1^*) \Theta(\alpha + \nu x_2^*)} e^{\frac{i \pi}{K} \nu (x_1^* + x_2^*)} = 0.
$$
This relation follows by taking the limit $\rho(x)$ in (\ref{rho-final}) as $x \to -\nu^{-1} (iK'+\alpha)$ or $\rho(x)$ in (\ref{rho-second}) as $x \to \nu^{-1} (iK'+\alpha)$:
\begin{equation*}
-1 = e^{\frac{i\pi \nu}{K} (x_1 + x_2 + x_1^* + x_2^*)} \frac{\Theta(\alpha - \nu x_1) \Theta(\alpha - \nu x_2) \Theta(\alpha - \nu x_1^*) \Theta(\alpha - \nu x_2^*)}{\Theta(\alpha + \nu x_1) \Theta(\alpha + \nu x_2) \Theta(\alpha + \nu x_1^*) \Theta(\alpha + \nu x_2^*)}.
\end{equation*}
Hence we get 
\begin{align*}
W = -e^{(s+s^*)x} \left[ \frac{\upsilon^2 e^{-\frac{i \pi \nu}{K} (x_1 + x_2)}}{C
	\Theta^2(\alpha - \nu x_1) \Theta^2(\alpha - \nu x_2)}  + \frac{(\upsilon^*)^2 e^{\frac{i \pi \nu }{K} (x_1^* + x_2^*)}}{C^*
	\Theta^2(\alpha + \nu x_1^*) \Theta^2(\alpha + \nu x_2^*)}  \right].
\end{align*}
The Wronskian $W$ is independent of $x$ if and only if $s + s^* = 0$.
\end{proof}

\begin{example}
As $\zeta \to 0$, we recover the exact solutions (\ref{exact-2}) and (\ref{exact-1}) from (\ref{p-gen}) and (\ref{p-gen-1}) by using 
\begin{align*}
x_1 &= \frac{iK' - \alpha}{\nu} = \frac{v}{2} + 2 i K' , \\
x_2 &= \frac{-iK' - \alpha}{\nu} = \frac{v}{2} , \\
x_1^*  &= \frac{iK' + \alpha}{\nu} = -\frac{v}{2}, \\
x_2^* &= \frac{-iK' + \alpha}{\nu} = -\frac{v}{2} - 2 i K' , 
\end{align*}
which agrees with the fact that roots of (\ref{square-root-1}) and (\ref{square-root-2}) approach to the double poles of $\phi^2(x)$ as $\zeta \to 0$. Since  $\mu \to 4 \zeta_1 \zeta_2 \zeta_3$ as $\zeta \to 0$, 
we use (\ref{s-final}) and obtain $s \to s_0$ as $\zeta \to 0$ with 
$$
s_0 = \zeta_1 - \zeta_2 - \zeta_3 - 2 Z'(\alpha) = -\nu \frac{H'(2\alpha)}{H(2 \alpha)},
$$
where we have used Lemma 4.2 in \cite{AP25} for the second equality. 
This agrees with (\ref{p1}) and (\ref{p2}) for $\delta = 2 \alpha$. 
We substitute values of $x_1$, $x_2$, $x_1^*$, $x_2^*$ into (\ref{p-gen}) and obtain 
$$
p_1(x) = -q_1(x) = \frac{\Theta(\nu x + \alpha)}{\Theta(\nu x - \alpha)} 
\frac{e^{\frac{i \pi \alpha}{K}}}{H^2(2 \alpha)} e^{-\nu x  \frac{H'(2\alpha)}{H(2\alpha)}},
$$
which is equivalent to (\ref{exact-2}) up to the scalar multiplication. Similarly, we substitute values of $x_1$, $x_2$, $x_1^*$, $x_2^*$ into (\ref{p-gen-1}) and obtain 
$$
p_2(x) = q_2(x) = \frac{\Theta(\nu x - \alpha)}{\Theta(\nu x + \alpha)} 
\frac{e^{\frac{i \pi \alpha}{K}}}{H^2(2 \alpha)} e^{\nu x  \frac{H'(2\alpha)}{H(2\alpha)}},
$$
which is equivalent to (\ref{exact-1}) up to the scalar multiplication.
\end{example} 

\subsection{Rational solutions}
\label{sec-5-1}

If $e_1 = e_2 = e_3$, then $\wp(x) = x^{-2}$, which implies from (\ref{phinova}) that 
\begin{equation}
\label{phi-rat}
\phi(x) = -\frac{1}{v} - \frac{4v}{4x^2 - v^2},
\end{equation}
with $v \in \mathbb{C}$ being the only parameter of the solution. 
It is clear that the rational solution (\ref{phi-rat}) is not important 
for applications since the solution is not real-valued if $v \in \mathbb{C}\backslash \mathbb{R}$ and is singular for real $x$ if $v \in \mathbb{R}$. This case is included for illustrations of expressions 
for roots $\{ \pm x_1,\pm x_2\}$ and $\{\pm x_1^*, \pm x_2^*\}$ 
of equations (\ref{square-root-1}) and (\ref{square-root-2}), respectively.

It follows from (\ref{pcW}) with $\wp(x) = x^{-2}$ that 
$$
c = \frac{6}{v^2}, \quad b = -\frac{4}{v^3}.
$$
By using (\ref{phi-rat}) in (\ref{third}) with the given values of $c$ and $b$, 
we also obtain 
$$
d = -\frac{3}{2v^4}.
$$
It follows from (\ref{parameterization}) with expressions for $b,c,d$ that $\zeta_1 = -v^{-1}$ and $\zeta_2 = \zeta_3 = v^{-1}$. By using parameter $a \in \mathbb{C}$, we get from $\zeta^2 = \wp(v) - \wp(a)$ and $\mu = -2 \wp'(a)$ that 
$$
\zeta^2 = \frac{a^2 - v^2}{a^2 v^2}, \quad \mu = \frac{4}{a^3}.
$$

The following lemma defines the roots $\{ x_1, x_2\}$ and $\{ x_1^*,x_2^*\}$
of Lemma \ref{prop-zeros} and \ref{prop-zeros-2} for the rational solutions. 

\begin{lemma}
	\label{lem-rational}
	The roots $\{ x_1, x_2\}$ and $\{ x_1^*,x_2^*\}$ of (\ref{square-root-1}) and (\ref{square-root-2}) for the rational solution (\ref{phi-rat}) are given by 
	$x_{1,2} = \frac{v}{2} y_{\pm}$ and $x_{1,2}^* = \frac{v}{2} y_{\pm}^*$, where 
	$\{ y_+, y_-\}$ and $\{ y_+^*, y_-^*\}$ are defined up to the sign choice by 
	\begin{equation}
	\label{y1-2}
	y_{\pm} = \sqrt{1 + \frac{4}{\pm \sqrt{1 + 2 (1-z^2) (1 + i z^{-1} \sqrt{1-z^2})}-1}}
	\end{equation}
	and
	\begin{equation}
	\label{y1-2-star}
	y_{\pm}^* = \sqrt{1 + \frac{4}{\pm \sqrt{1 + 2 (1-z^2) (1 - i z^{-1} \sqrt{1-z^2})}-1}},
	\end{equation}
	where $z = v \zeta$.
\end{lemma}

\begin{proof}
Equation (\ref{square-root-1}) is rewritten explicitly as 
\begin{equation}
\label{eq-1-app}
\frac{16 v^2}{(4 x^2 - v^2)^2} + \frac{8}{4 x^2 - v^2} - \frac{2}{a^2} - \frac{2i}{a^3 \zeta} = 0.
\end{equation}
In order to scale $v \in \mathbb{C}$ out, we use the scaled variables 
\begin{equation}
\label{eq-scaled-app}
x = \frac{v}{2} y, \quad \zeta = \frac{z}{v}, \quad a = \frac{v}{\sqrt{1-z^2}}.
\end{equation}
Equation (\ref{eq-1-app}) in scaled variables (\ref{eq-scaled-app}) 
can be rewritten in the form 
\begin{equation}
\label{eq-2-app}
\frac{16}{(y^2 - 1)^2} + \frac{8}{y^2 - 1} - 2 (1-z^2) \left(1 + i z^{-1} \sqrt{1-z^2} \right) = 0,
\end{equation}
from which we obtain two pairs of roots $\{ \pm y_+, \pm y_- \}$ in the form (\ref{y1-2}). Similarly, equation (\ref{square-root-2}) can be rewritten in the form
\begin{equation}
\label{eq-3-app}
\frac{16}{(y^2 - 1)^2} + \frac{8}{y^2 - 1} - 2 (1-z^2) \left(1 - i z^{-1} \sqrt{1-z^2} \right) = 0,
\end{equation}
from which we obtain two pairs of roots $\{ \pm y_+^*, \pm y_-^* \}$ in the form (\ref{y1-2-star}).
\end{proof}

\begin{figure}[htb!]
	\includegraphics[width=0.8\textwidth]{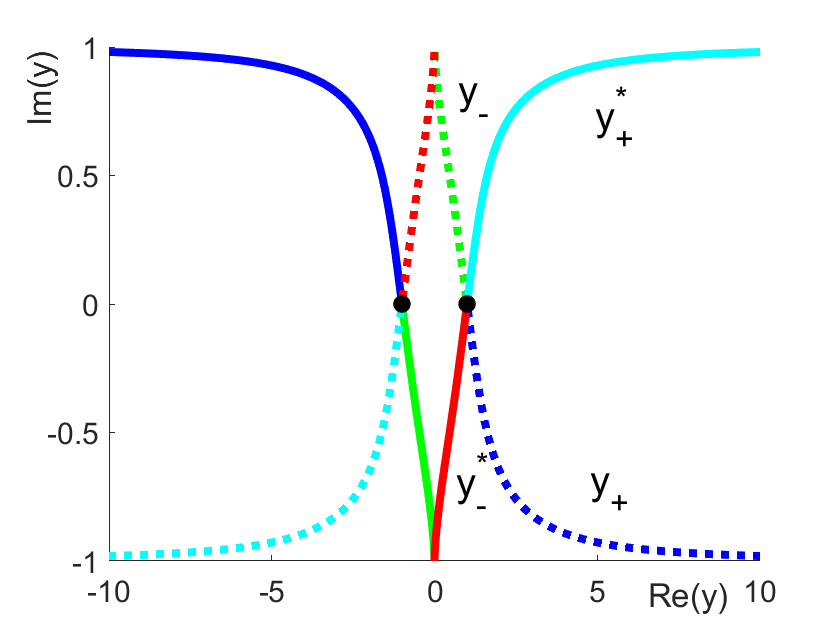}
	\caption{Pairs of roots (\ref{y1-2}) and (\ref{y1-2-star}) on the complex $y$-plane parameterized by $z \in [0,1]$. The solid lines correspond to roots of (\ref{eq-2-app}) and (\ref{eq-3-app}), which are simultaneously 
		roots of (\ref{eq-4-app}). The dashed lines correspond to 
	roots of (\ref{eq-2-app}) and (\ref{eq-3-app}), which are simultaneously 
	roots of (\ref{eq-5-app}).}
	\label{fig-roots}
\end{figure}

Let us fix $z \in (0,1)$. Figure \ref{fig-roots} show 
pairs of roots $\pm y_+$, $\pm y_-$, $\pm y_+^*$, $\pm y_-^*$ on the 
complex $y$-plane parameterized by $z \in (0,1)$. The roots converge to $\pm 1$ as $z \to 0$ (shown by black dots). Two pairs of roots converge to $i$ as $z \to 1$ and two pairs of roots diverge to infinity as $z \to 1$. 
The solid lines show the roots of (\ref{eq-2-app}) and (\ref{eq-3-app}), 
which are simultaneously roots of 
\begin{equation}
\label{eq-4-app}
\frac{8 i zy}{(y^2-1)^2} + \frac{4 z^2}{y^2-1} - (1-z^2) = 0,
\end{equation}
which is obtained from (\ref{zero-denominator}) with the scaling (\ref{eq-scaled-app}). Similarly, the dashed lines show the sign-reflected 
roots, which are roots of 
\begin{equation}
\label{eq-5-app}
-\frac{8 i zy}{(y^2-1)^2} + \frac{4 z^2}{y^2-1} - (1-z^2) = 0,
\end{equation}
obtained from (\ref{zero-denominator-recipr}) with the scaling (\ref{eq-scaled-app}). The appropriate roots of (\ref{y1-2}) and (\ref{y1-2-star}), which are simultaneously roots of (\ref{eq-4-app}) and (\ref{eq-5-app}) are identified with the help of computer computations. 

Thus, based on the numerical data and 
the scaling transformation (\ref{eq-scaled-app}), we have 
$$
x_1 = -\frac{v}{2} y_+, \quad 
x_2 = -\frac{v}{2} y_-, \quad 
x_1^* = \frac{v}{2} y_+^*, \quad 
x_2^* = \frac{v}{2} y_-^*,
$$
from which the relation $y_{\pm}^* = \overline{y}_{\pm}$ suggests 
that $x_{1,2}^* = -\overline{x}_{1,2}$. We emphasize that 
the roots $\{ x_+,x_-\}$ and $\{ x_+^*,x_-^*\}$ obtained from 
(\ref{y1-2}) and (\ref{y1-2-star}) depends on the spectral 
parameter $\zeta$ via the square root singularities 
at $z = 0$ and $z = \pm 1$. 

\subsection{Hyperbolic solutions} 
\label{sec-5-2}

If $0 < \zeta_3 = \zeta_2 < \zeta_1$, we have $k = 1$ from (\ref{parameters-nu-k}), which corresponds to the solitary wave rather than 
the periodic wave of the mKdV equation (\ref{mkdv}). Since
$K = \infty$ and $K' = \frac{\pi}{2}$, the fundamental rectangle $[-K,K] \times [-iK',iK']$ becomes a horizontal strip of the width $\pi$ in the vertical direction. 

The following lemma defines the roots $\{ x_1, x_2\}$ and $\{ x_1^*,x_2^*\}$
of Lemma \ref{prop-zeros} and \ref{prop-zeros-2} for the hyperbolic solutions. 

\begin{lemma}
	\label{lem-hyperbolic}
	The roots $\{ x_1, x_2\}$ and $\{ x_1^*,x_2^*\}$ of (\ref{square-root-1}) and (\ref{square-root-2}) for the hyperbolic solution (\ref{form-2}) 
	with $0 < \zeta_3 = \zeta_2 < \zeta_1$ are given by 
	$x_{1,2} = \nu^{-1} z_{1,2}$ and $x_{1,2}^* = \nu^{-1} z_{1,2}^*$, where 
	$\{ z_1, z_2\} \in \R \times \left[-\frac{i \pi}{2},\frac{i \pi}{2}\right]$ and $\{ z_1^*, z_2^*\} \in \R \times \left[-\frac{i \pi}{2},\frac{i \pi}{2}\right]$ are defined up to the sign choice by 
	roots of 
\begin{align}
\cosh(2z) = \frac{-\zeta_2^2 (\zeta_1^2-\zeta^2) + i \zeta (\zeta_1^2 - \zeta_2^2) \sqrt{\zeta_1^2 - \zeta^2} \pm (\zeta_1^2-\zeta_2^2) \sqrt{ \zeta^2 (2 \zeta_2^2-\zeta_1^2) - 2 i \zeta \zeta_2^2 \sqrt{\zeta_1^2 - \zeta^2}}}{\zeta_1 \zeta_2 (\zeta_2^2 - \zeta^2)}
\label{roots-sech}
\end{align}
	and
\begin{align}
\cosh(2z) = \frac{-\zeta_2^2 (\zeta_1^2-\zeta^2) - i \zeta (\zeta_1^2 - \zeta_2^2) \sqrt{\zeta_1^2 - \zeta^2} \pm (\zeta_1^2-\zeta_2^2) \sqrt{ \zeta^2 (2 \zeta_2^2-\zeta_1^2) + 2 i \zeta \zeta_2^2 \sqrt{\zeta_1^2 - \zeta^2}}}{\zeta_1 \zeta_2 (\zeta_2^2 - \zeta^2)},
\label{roots-sech-star}
\end{align}
where  $\zeta \in (0,\zeta_2)$.
\end{lemma}

\begin{proof}
	It follows from the proof of Lemma \ref{prop-zeros} that 
	equation (\ref{square-root-1}) can be written as equation (\ref{root-elliptic}), which is further rewritten as 
$$
4 i \zeta (\zeta_1^2 - \zeta_2^2)\frac{\zeta_2^2 \cosh^2(2z) + \zeta_1^2 - 2 \zeta_2^2}{(\zeta_1 + \zeta_2 \cosh(2z))^2} + 4 i \zeta (\zeta^2 - \zeta_1^2) + \mu = 0. 
$$
This equation can be expanded as 
$$
\zeta_2^2 \left(\zeta^2 - \zeta_2^2 + \frac{\mu}{4i\zeta} \right) \cosh^2(2z) + 2 \zeta_1 \zeta_2 \left(\zeta^2 - \zeta_1^2 + \frac{\mu}{4i\zeta} \right) \cosh(2z) + \zeta^2 \zeta_1^2 - 3 \zeta_1^2 \zeta_2^2 + 2 \zeta_2^4 + \frac{\mu \zeta_1^2}{4i\zeta}  = 0.
$$
Since $\mu^2 = -16 (\zeta^2 - \zeta_1^2) (\zeta^2 - \zeta_2^2)^2$, multiplying by 
$\left( \zeta^2 - \zeta_2^2 - \frac{\mu}{4i\zeta} \right)$ yields the following equation after straightforward computations:
\begin{align*}
& \left( \zeta_1 \zeta_2 (\zeta^2 - \zeta_2^2) \cosh(2z) + \zeta_2^2 (\zeta^2 - \zeta_1^2) + \frac{\mu \zeta^2 (\zeta_1^2 - \zeta_2^2)}{4 i \zeta (\zeta^2 - \zeta_2^2)} \right)^2 \\
& \qquad + (\zeta_1^2 - \zeta_2^2)^2 \left( \zeta^2 (\zeta_1^2 - 2 \zeta_2^2) + \frac{2 \mu \zeta^2 \zeta_2^2}{4 i \zeta (\zeta^2 - \zeta_2^2)} \right) = 0.
\end{align*}
If $\zeta \in (0,\zeta_2)$, then $\mu = 4 (\zeta_2^2 - \zeta^2) \sqrt{\zeta_1^2 - \zeta^2}$ follows from (\ref{omega-final}). Solving the previous equation for $\cosh(2z)$ yields two solutions given by (\ref{roots-sech}). Similarly, 
equation (\ref{square-root-2}) can be solved with two solutions given by (\ref{roots-sech-star}).
\end{proof}

\begin{figure}[htb!]
	\includegraphics[width=0.8\textwidth]{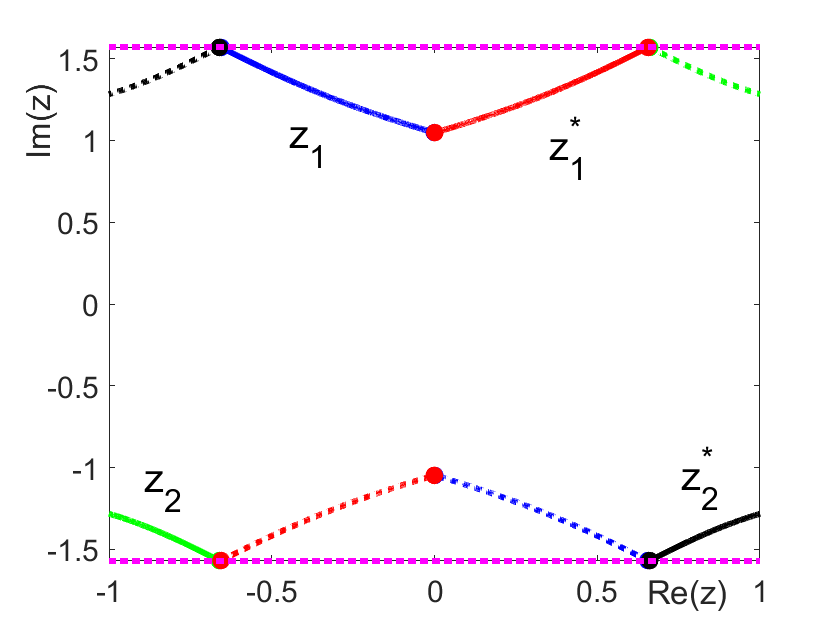}
	\caption{Roots $\{ \pm z_1, \pm z_2\}$ of (\ref{roots-sech}) and $\{ \pm z_1^*, \pm z_2^*\}$ of (\ref{roots-sech-star}) on the complex plane for $\zeta \in (0,\zeta_2)$ for  $\zeta_1 = 1$ and $\zeta_2 = \zeta_3 = 0.5$.}
	\label{fig-hyperbolic}
\end{figure}

Figure \ref{fig-hyperbolic} displays roots $\{ \pm z_1, \pm z_2\}$ and $\{ \pm z_1^*, \pm z_2^* \}$ of (\ref{roots-sech}) and (\ref{roots-sech-star}) for $\zeta_1 = 1$ and $\zeta_2 = \zeta_3 = 0.5$. The roots are parameterized by $\zeta \in (0,\zeta_2)$. The big dots of the same color show roots at $\zeta = \zeta_2$ and $\zeta = 0$. Dashed horizontal lines in magenta shows the boundary of the vertical strip at $\pm iK' = \pm \frac{\pi i}{2}$. For $\zeta = 0$, four pairs of roots are located at $\pm i K' \pm \alpha$. As $\zeta \to \zeta_2$, four roots diverge to infinity since $K = \infty$ and two pairs of roots coalesce on $i \mathbb{R}$.

The solid lines show the roots of (\ref{roots-sech}) and (\ref{roots-sech-star}) which are simultaneously roots of (\ref{zero-denominator}), that is, roots labeled as $\{ z_1,z_2,z_1^*,z_2^*\}$. Consequently, the dashed lines show the reflected roots $\{-z_1,-z_2,-z_1^*,-z_2^*\}$ which are roots of (\ref{roots-sech}) and (\ref{roots-sech-star}) which are simultaneously roots of (\ref{zero-denominator-recipr}). Again, the roots are identified with the help of computer computations. We can see from Figure \ref{fig-hyperbolic} that the roots satisfy the symmetry
$z_{1,2}^* = -\bar{z}_{1,2}$, which suggests 
the same relation $x_{1,2}^* = -\overline{x}_{1,2}$ as in the case of rational solutions. 

\vspace{0.2cm}

{\bf Acknowledgement.}  The first author thanks M. Bertola, D. Korotkin, and L. Ling for many discussions related to this project. The main breakthrough for this project was obtained during the Third Joint Alabama-Florida Conference on Differential Equations, Dynamical Systems and Applications at University of Alabama Birmingham (UAB), May 20-22, 2025. Some computations in Section \ref{sec-5} were obtained earlier, in collaboration with L. K. Arruda during the previous project \cite{AP25}. The work of D. E. Pelinovsky was supported in part
by the National Natural Science Foundation of China (grant no. 12371248).

\vspace{0.2cm}

{\bf Conflict of interests.}  The authors declare no conflict of interests. 

\vspace{0.2cm}

{\bf Data availability.}  Data is available upon a reasonable request.

\vspace{0.2cm}

{\bf Funding statements.}  The authors report no funding for this work.

\end{document}